\documentclass[letterpaper, 10 pt, conference]{ieeeconf}
\IEEEoverridecommandlockouts                             
\usepackage{amsfonts}
\usepackage{amsmath} 
\usepackage{amssymb}  

\usepackage{float}

\usepackage{amsthm}

\theoremstyle{definition}
\newtheorem{proposition}{Proposition}[section]

\usepackage{hyperref}
\hypersetup{
    colorlinks=true,
    linkcolor=blue,
    filecolor=magenta,      
    urlcolor=cyan,
}

\usepackage{tikz}
\usetikzlibrary{calc,patterns,decorations.pathmorphing,decorations.markings}
\usetikzlibrary{shapes,arrows}
\usepackage{verbatim}

\tikzstyle{block} = [draw, fill=blue!20, rectangle, 
    minimum height=3em, minimum width=6em]
\tikzstyle{sum} = [draw, fill=blue!20, circle, node distance=1cm]
\tikzstyle{input} = [coordinate]
\tikzstyle{output} = [coordinate]
\tikzstyle{pinstyle} = [pin edge={to-,thin,black}]

\usetikzlibrary{positioning}
\usetikzlibrary{math}

\usepackage{tikz}
\usetikzlibrary{shapes,arrows,calc,positioning}
\tikzstyle{bigblock} = [draw, fill=blue!20, rectangle, 
    minimum height=6em, minimum width=8em]
\tikzstyle{medblock} = [draw, fill=blue!20, rectangle, 
    minimum height=4em, minimum width=4em]    
\tikzstyle{mux} = [draw, fill=black!20, rectangle, 
    minimum height=5em, minimum width=0.1em]    
\tikzstyle{smallblock} = [draw, fill=blue!20, rectangle, 
    minimum height=3em, minimum width=4em]
\tikzstyle{sum} = [draw, fill=blue!20, circle, node distance=1cm]
\tikzstyle{signal} = [coordinate]
\tikzstyle{pinstyle} = [pin edge={to-,thin,black}]
\tikzstyle{block} = [draw, fill=blue!20, rectangle, 
    minimum height=3em, minimum width=6em]
\tikzstyle{blockS} = [draw, fill=blue!20, rectangle, 
    minimum height=3em, minimum width=4em]    
\tikzstyle{input} = [coordinate]
\tikzstyle{output} = [coordinate]
\usetikzlibrary{matrix}

\newcommand{\bc}{\begin{center}}
\newcommand{\ec}{\end{center}}
\newcommand{\benum}{\begin{enumerate}}
\newcommand{\eenum}{\end{enumerate}}
\newcommand{\nn}{\nonumber}
\newcommand{\matl}{\left[ \begin{array}}
\newcommand{\matr}{\end{array} \right]}
\newcommand{\matls}{\left[ \begin{smallmatrix}}
\newcommand{\matrs}{\end{smallmatrix} \right]}
\newcommand{\isdef}{\stackrel{\triangle}{=}}

\newcommand{\vect}[1]{\overset{\rightharpoonup}{#1}}

\newcommand{\rmA}{{\rm A}}
\newcommand{\rmB}{{\rm B}}

\newcommand{\rmD}{{\rm D}}
\newcommand{\rmE}{{\rm E}}
\newcommand{\rmF}{{\rm F}}

\newcommand{\rmI}{{\rm I}}

\newcommand{\rmP}{{\rm P}}
\newcommand{\rmQ}{{\rm Q}}

\newcommand{\rmT}{{\rm T}}

\newcommand{\rmc}{{\rm c}}
\newcommand{\rmd}{{\rm d}}

\newcommand{\rmi}{{\rm i}}

\newcommand{\rmp}{{\rm p}}

\newcommand{\BBR}{{\mathbb R}}

\newcommand{\SO}{{\mathcal O}}

\newcommand{\SQ}{{\mathcal Q}}

\newcommand{\shiftq}{{\textbf{\textrm{q}}}}

\newcommand{\framedotE}[1]{\stackrel{\rmE\bullet}{#1}}

\newcommand{\resolvedin}[1]{{\big|_{\rm #1}}}
\renewcommand{\resolvedin}[2]{{\big|_{\rm #1}^{\rm #2}}}

\newcommand{\ihat}{ {\hat \imath}}
\newcommand{\jhat}{ {\hat \jmath}}
\newcommand{\khat}{ {\hat k}}

\newcommand{\rotation}[2]{ {\substack{#1 \\ \boldsymbol{\longrightarrow} \\ #2}} }

\usepackage[style=ieee ,sorting=none,maxbibnames=99,giveninits, doi=false]{biblatex}
\usepackage{fancyhdr} 
\title{An A Quadcopter Autopilot Based on an Adaptive Digital PID Controller}

\title{Adaptive Digital PID Control of a Quadcopter}

\title{Retrospective-Cost-Based Adaptive Digital PID Control of a Quadcopter}

\title{A Retrospective-Cost-Based Adaptive Digital PID   Quadcopter Autopilot}

\title{Adaptive Digital PID Control of a Quadcopter with Unknown Dynamics}

\title{One-Shot Learning for a Quadcopter Autopilot}

\title{An adaptive digital autopilot for Multicopters}

\title{\LARGE \bf Experimental Implementation of an Adaptive Digital Autopilot }

\author{
    Ankit Goel, 
    Juan Augusto Paredes, 
    Harshil Dadhaniya, 
    Syed Aseem Ul Islam, 
    Abdulazeez Mohammed Salim, \\
    Sai Ravela, 
    Dennis Bernstein%
\thanks{This research was supported in part by the Office of Naval Research under grant N00014-19-1-2273.}
\thanks{Ankit Goel, Juan Augusto Paredes, Harshil Dadhaniya, Syed Aseem Ul Islam, and Dennis Bernstein are with the Department of Aerospace Engineering, University of Michigan, Ann Arbor, MI 48109.
{\tt\small ankgoel,jparedes,hdadhani,aseemisl,}
{\tt \small dsbaero@umich.edu}
}
\thanks{Abdulazeez Mohammed Salim is with the Department of Aeronautics and Astronautics, MIT, Cambridge, MA 02139.
{\tt\small azez@mit.edu}} 
\thanks{Sai Ravela is with the Department of Earth, Atmospheric, and Planetary Sciences, MIT, Cambridge, MA 02139.
{\tt\small ravela@mit.edu}} 
}

\date{August 2019}

\begin{document}

\maketitle

\begin{abstract}
    This paper develops an adaptive digital autopilot for quadcopters and presents experimental results.
    The adaptive digital autopilot is constructed by augmenting the PX4 autopilot control system architecture with  adaptive digital control laws based on retrospective cost adaptive control (RCAC).  
    In order to investigate the performance of the adaptive digital autopilot, the default gains of the fixed-gain autopilot are scaled by a small factor, which severely degrades its performance.
    This scenario thus provides a venue for determining the ability of the adaptive digital autopilot to compensate for the detuned fixed-gain autopilot.
    The adaptive digital autopilot is tested in simulation and physical flight tests, and the resulting performance improvements are examined.  
    
\end{abstract}

\section{Introduction}

Multicopters are ubiquitous and are increasingly being used in a wide range of applications, such as sports broadcasting, wind-turbine inspection and agricultural monitoring \cite{chang2016development,anweiler2017multicopter,andaluz2015nonlinear, schafer2016multicopter,stokkeland2015autonomous,juan2017}.
In the most common quadcopter (multicopters with four propellers) configuration, the rotation directions and spin rates of the four motors provides thrust for translational motion as well as moments for attitude control.
However, the nonlinear and unstable dynamics makes autonomous operation of quadcopters a very challenging problem. 
The control system of a multicopter is thus finely tuned and tailored to the geometry and mass properties of the vehicle.
In fact, the open-source autopilots PX4 and ArduPilot contain finely-tuned controller gains for many commercially available unmanned aerial vehicle (UAV) configurations \cite{meier2015px4,ardupilot}.
In some applications, however, vehicle properties are subject to changes due to hardware alterations, such as airframe, payload, sensors, and actuators, and environmental conditions, such as wind speed and air density,
which occur especially in experimental and field operations.
In these cases, there is no guarantee that stock autopilot gains, tuned to perform well in a particular setting, will perform in an acceptable manner in off-nominal environments.
Along the same lines, unanticipated and unknown changes that occur during flight due to failure or damage may significantly degrade the performance of the stock autopilot.

With this motivation in mind, the present paper develops an adaptive digital autopilot for quadcopters with unknown dynamics.
To do this, the PX4 autopilot architecture is modified so that the feedback and feedforward P and PID controllers are augmented with adaptive control laws based on retrospective cost adaptive control (RCAC) \cite{rahmanCSM2017}.
In particular, each controller in PX4 is augmented by an adaptive digital controller as described in \cite{rezaPID}.
The adaptive digital controller is based on recursive least squares (RLS), and thus involves the update of a matrix of size up to $4\times4$ (3 PID gains + feedforward gain) at each time step, which makes it suitable for meeting the time constraints imposed by the real-time embedded systems used to implement the PX4 autopilot.

Learning algorithms have been previously implemented in UAV autopilots.
In \cite{dydek2013MRAC}, the performance of model reference adaptive control (MRAC) was tested in a quadcopter under flight failure conditions, where a propeller was cut mid-flight.
Previous knowledge of quadcopter dynamics was used to design the algorithm based on MRAC.
In contrast, the adaptive algorithm proposed in this paper does not require any previous knowledge regarding the system dynamics.
Fuzzy neural network based sliding mode control was used in \cite{kayacan2017learning} to control a fixed-wing aircraft flying under wind disturbances and, simultaneously, learn the inverse dynamics of the plant model.
However, the autopilot required the gains from the proportional controllers to be appropriately initialized to allow sufficient time for learning.
In this paper, in contrast, the adaptive digital autopilot controller coefficients are all initialized at zero and, thus, don't require any conditioning.
Retrospective-cost based PID controllers were used in the attitude controller in \cite{ansari2018retrospective}, and were applied with fixed hyperparameter tuning to a quadcopter, a fixed-wing aircraft, and a VTOL aircraft in a simulation environment.
In \cite{dai_quad_2014}, retrospective-cost adaptive control law was used to compensate for the payload mass uncertainty.
The present paper extends this work by augmenting both the attitude and position controllers in the PX4 autopilot with adaptive controllers and testing it in both simulation and experimental settings.

The contribution of the paper is the development of an adaptive digital autopilot and its experimental demonstration.
The finely tuned stock autopilot is degraded by scaling down all fixed-gain controllers and the adaptive digital autopilot is then used to
%
%
recover the stock autopilot's performance. 
The improvements are demonstrated through simulation and flight tests results. 

The paper is organized as follows:
Section \ref{sec:QuadDyn} summarizes the quadcopter dynamics and the notation used in this paper.
Section \ref{sec:PX4_autopilot} presents the control system architecture implemented in the PX4 autopilot in detail.
Section \ref{sec:PID_Algo} presents the retrospective cost based adaptive control algorithm.
Section \ref{sec:adaptiveAugmentation} describes the augmentation of the PX4 autopilot.
Section \ref{sec:flight_tests} presents simulation tests and flight test results of the stock PX4 and the adaptive PX4 autopilot. 
%
%
Finally, section \ref{sec:conclusions} concludes the paper with a summary and future research directions.

\section{Quadcopter Dynamics}
\label{sec:QuadDyn}

This section describes the quadcopter dynamics and the notation used in this paper.
The Earth frame and quadcopter body-fixed frame are denoted by the row vectrices $\rm{F}_{\rm{E}} = \begin{bmatrix} \hat{\imath}_{\rm{E}} & \hat{\jmath}_{\rm{E}} & \hat{k}_{\rm{E}} \end{bmatrix}$ and $\rm{F}_{\rm{Q}} = \begin{bmatrix} \hat{\imath}_{\rm{Q}} & \hat{\jmath}_{\rm{Q}} & \hat{k}_{\rm{Q}} \end{bmatrix}$, respectively. 
We assume that $\rm{F}_{\rm{E}}$ is an inertial frame and the Earth is flat. 
The origin $w$ of $\rm{F}_{\rm{E}}$ is any convenient point fixed on the Earth.
The axes $\hat\imath_{\rm{E}}$ and $\hat\jmath_{\rm{E}}$ are horizontal, while the axis $\hat{k}_{\rm{E}}$ points downwards.
%
%
$\rm{F}_{\rm{Q}}$ is defined with $\hat\imath_{\rm{Q}}$ and $\hat\jmath_{\rm{Q}}$ in the plane of the rotors, and $\hat{k}_{\rm{Q}}$ points downwards, that is, $\hat{k}_{\rm{Q}} = \hat\imath_{\rm{Q}} \times \hat\jmath_{\rm{Q}}$. 
Assuming that $\hat\imath_{\rm{E}}$ points North and $\hat\jmath_{\rm{E}}$ points  East, it follows that the Earth frame is a local NED frame.
%
%
The quadcopter frame $\rm F_{\rm Q}$ is obtained by applying a 3-2-1 sequence of Euler-angle rotations to the Earth frame $\rm F_{\rm E}$, where $\Psi,$ $\Theta,$ and $\Phi$ denote the azimuth, elevation, and bank angles, respectively.
The frames $\rm F_Q $ and $\rm F_E$ are thus related by 
\begin{align}
    \rmF_\rmE
        \mathop{\longrightarrow}^{\Psi}_{3} 
    \rmF_\rmA
        \mathop{\longrightarrow}^{\Theta}_{2} 
    \rmF_\rmB
        \mathop{\longrightarrow}^{\Phi}_{1} 
    \rmF_\rmQ.
\end{align}


The translational equations of motion of the quadcopter are given by
\vspace{-1em}
\begin{align}
    m \overset{\rmE \bullet \bullet} {\vect r}_{\rmc/w} 
        &=
            m \vect g + \vect f_\rmc,
    \label{eq:N2L_pos}
\end{align}
where 
$m$ is the mass of the quadcopter,
$\rmc$ is the center-of-mass of the quadcopter, 
$\vect r_{\rmc/w} $ is the physical vector representing the position of the center-of-mass $\rmc$ of the quadcopter relative to $w$,
$\vect g = g \hat k_\rmE,$
and $\vect f_\rmc$ is the total force applied on $\rmc.$
Note that the position of the center-of-mass $\rmc$ of the quadcopter relative to $w$ resolved in the Earth frame $\rm F_E$ is ${\vect r}_{\rmc/w} \resolvedin{E}{}$ and the velocity of $\rmc$ relative to $w$ with respect to the Earth frame $\rm F_E$ resolved in the Earth frame $\rm F_E$ is $\framedotE {\vect r}_{\rmc/w} \resolvedin{E}{}.$
%

%
The rotational equations of motion of the quadcopter in coordinate-free form are given by
\begin{align}
    \vec J_{{\rm Q}/\rmc} 
    \overset{\rm\rmE \bullet } { \vect \omega}_{\rmQ/\rmE} +
    \vect \omega_{\rmQ/\rmE} \times \vec J_{{\rm Q}/\rmc} \vect \omega_{\rmQ/\rmE}
        &=
            \vect M_{{\rm Q}/\rmc},
    \label{eq:EulersEqn}
\end{align}
where 
$\vec J_{{\rm Q}/\rmc}$
is the inertia tensor of the quadcopter,  
$\vect M_{{\rm Q}/\rmc}$ is the moment applied to the quadcopter relative to $\rmc$, and 
$\vect \omega_{\rmQ/\rmE} 
    = 
        P \hat \imath _{\rm Q} + 
        Q \hat \jmath _{\rm Q} + 
        R \hat k _{\rm Q}$
is the angular velocity of frame $\rmF_\rmQ$ relative to the inertial Earth frame $\rmF_\rmE.$ 
%
%
%
Furthermore, the attitude of the frame $\rm F_Q$, represented by $\SO_{\rm Q/E}$, satisfies
\begin{align}
    \dot \SO_{\rm Q/E} 
        =
            -\vect \omega_{\rm Q/E} \resolvedin{Q}{\times}
            \SO_{\rm Q/E} ,
\end{align}
where
$\vect \omega_{\rm Q/E} \resolvedin{Q}{\times}$
is a $3 \times 3$ skew-symmetric, cross-product matrix.

\section{PX4 Autopilot}
\label{sec:PX4_autopilot}
In this section, the control system implemented in the stock PX4 autopilot is described. 
The control system consists of a mission planner and two nested loops as shown in Figure \ref{fig:PX4_autopilot_nested_loop}.
The mission planner generates the position, velocity, azimuth, and azimuth rate setpoints from the user-defined waypoints.
%
The outer loop consists of the \textit{position controller} whose inputs are the position setpoint $\vect r_{c/w}\resolvedin{E}{\rm sp}$ and velocity setpoints $\framedotE{\vect r}_{c/w}\resolvedin{E}{\rm sp,ff}$ as well as the measured position $\vect r_{c/w}\resolvedin{E}{\rm meas}$ and measured velocity $\framedotE{\vect r}_{c/w}\resolvedin{E}{\rm meas}$ of the quadcopter.
The output of the position controller is the thrust vector setpoint $\vect f_\rmc\resolvedin{E}{\rm sp}$.
The inner loop consists of the \textit{attitude controller} whose inputs are the thrust vector setpoint, the azimuth setpoint $\Psi_{\rm sp},$ and azimuth rate setpoints $\dot \Psi_{\rm sp,ff}$, as well as the measured attitude $q_{\rm Q/E}^ {\rm meas}$ and the angular velocity measured in the body-fixed frame $\vect \omega_{\rm Q/E} \resolvedin{Q}{\rm meas}$. 
The output of the attitude controller is the moment setpoint $\vect M_{\SQ/\rmc}\resolvedin{Q}{\rm sp}$.
The magnitude of the thrust vector and the moment vector uniquely determine the rotation rate of the four propellers. 

\vspace{-1.25em}
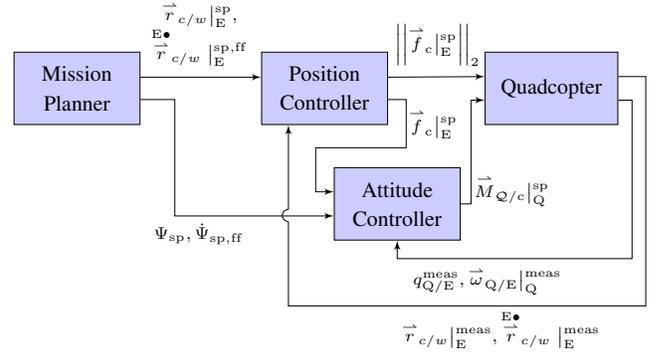
\begin{figure}[h]
    \centering
    \resizebox{\columnwidth}{!}{
    \begin{tikzpicture}[auto, node distance=2cm,>=latex',text centered]
    
        \node [smallblock, minimum height=3em, text width=1.6cm] (Mission) {\small Mission Planner};
        \node [smallblock, minimum height=3em, right = 5em of Mission, text width=1.6cm] (Pos_Cont) {\small Position Controller};
        \node [smallblock, minimum height=3em, below right = 1.75em and -2.25 em of Pos_Cont,text width=1.6cm] (Att_Cont) {\small Attitude Controller};
        \node [smallblock, minimum height=3em, minimum width = 5.5em,  right = 4em of Pos_Cont] (Quadcopter) {\small Quadcopter};
        
        \draw [->] (Mission.10) -- node[above, xshift = -0.05 em, yshift = 0.1em]{\scriptsize$\begin{array}{c}\vect r_{c/w}\resolvedin{E}{\rm sp}, \\ \framedotE{\vect r}_{c/w}\resolvedin{E}{\rm sp,ff} \end{array}$} (Pos_Cont.170);
        \draw[->] (Mission.350) -- +(.5,0) |- node[below, xshift = 1 em, yshift = 0.05em]{\scriptsize$\Psi_{\rm sp}, \dot \Psi_{\rm sp,ff}$} ([yshift = -0.5em]Att_Cont.180);
        \draw [->] (Pos_Cont.-10) -| ([xshift = 0.75em, yshift = -2em]Pos_Cont.-10) -| node  [xshift = 3.5em, yshift = 1em] {\scriptsize$\vect f_\rmc\resolvedin{E}{\rm sp}$} ([xshift = -0.75em, yshift = 0.5em]Att_Cont.west) -- ([yshift = 0.5em]Att_Cont.west);
        \draw [->] (Pos_Cont.10) -- node [above, xshift = 0.05em, yshift = 0.1 em] {\scriptsize$\left|\left|\vect f_\rmc\resolvedin{E}{\rm sp}\right|\right|_2$} (Quadcopter.170);
        \draw [->] (Att_Cont.0) -- +(0.15,0) |- node [below, xshift = 1.7 em, yshift = -2.75 em]{\scriptsize$\vect M_{\SQ/\rmc}\resolvedin{Q}{\rm sp}$}(Quadcopter.190);
        \draw [-] (Quadcopter.10) -- +(.4,0) |- node[below,xshift = -6em, yshift = 0.2em]{\scriptsize $\vect r_{c/w} \resolvedin{E} {\rm meas}, \framedotE{\vect r}_{c/w}\resolvedin{E} {\rm meas}$} ([xshift = -1.5em, yshift = -7.5 em]Pos_Cont.south) -- ([xshift = -1.5em, yshift = -4 em]Pos_Cont.south);
        \draw ([xshift = -1.5em, yshift = -4 em]Pos_Cont.south) arc (270:90:0.25em);
        \draw [->] ([xshift = -1.5em, yshift = -3.5 em]Pos_Cont.south) -- ([xshift = -1.5em]Pos_Cont.south);
        \draw [->] (Quadcopter.350) -- +(0.2,0) |- node[below,xshift = -6em]{\scriptsize $q_{\rm Q/E}^ {\rm meas}, \vect \omega_{\rm Q/E} \resolvedin{Q}{\rm meas}$} ([yshift = -0.75 em]Att_Cont.south) -- (Att_Cont.south);

    \end{tikzpicture}
    }
    \caption{\footnotesize PX4 autopilot architecture. }
    \label{fig:PX4_autopilot_nested_loop}
\end{figure}

The position controller consists of two cascaded linear controllers as shown in Figure \ref{fig:PX4_autopilot_outer_loop}.
The first controller $G_r$ consists of a proportional controller and a feedforward controller.
The velocity setpoint is thus given by 
\begin{align}
    \framedotE{\vect r}_{c/w}\resolvedin{E}{\rm sp}
        =
            K_r
            z_r
            +
            \framedotE{\vect r}_{c/w}\resolvedin{E}{\rm sp,ff},
\end{align}
where 
$z_r \isdef \vect r_{c/w}\resolvedin{E}{\rm sp} -
            \vect r_{c/w}\resolvedin{E}{\rm meas}, $
$K_r$ is a $3\times3$ diagonal matrix,  
and 
$\framedotE{\vect r}_{c/w}\resolvedin{E}{\rm sp,ff}$ is the
%
%
feedforward velocity setpoint specified by the mission planner. 
Note that diagonal entries of $K_r$ are the tuning gains.

The second controller $G_v$ consists of three decoupled PID controllers. 
The force setpoint $\vect f_\rmc \resolvedin{E}{{\rm  sp}}$ is thus given by
\begin{align}
    \vect f_\rmc \resolvedin{E}{{\rm  sp}}
        =
            G_v 
            z_v,
\end{align}
where 
$z_v \isdef \framedotE{\vect r}_{c/w}\resolvedin{E}{\rm sp}
                - \framedotE{\vect r}_{c/w}\resolvedin{E}{\rm meas},$
\begin{align}
    G_v 
        = 
            K_{v,\rmP} + 
            \dfrac{K_{v,\rmI}}{\shiftq-1} + 
            K_{v,\rmD}\left(1-\frac{1}{\shiftq}\right),
\end{align}
and 
$K_{v,\rmP}, K_{v,\rmI},$ and  $K_{v,\rmD}$ are $3\times3$ diagonal matrices; and $\shiftq$ is the forward-shift operator. 
Note that diagonal entries of $K_{v,\rmP}, K_{v,\rmI},$ and  $K_{v,\rmD}$ are the tuning gains.



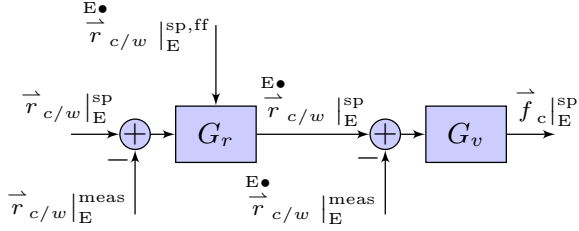
\begin{figure}[h]
    \centering
    \resizebox{0.95\columnwidth}{!}{
    \begin{tikzpicture}[auto, node distance=2cm,>=latex']
        \node (Ref_traj) {};
    	\node [sum, right = 1.5 em of Ref_traj] (sum1) {};
    	\node[draw = white] at (sum1.center) {$+$};
    	\node [smallblock, minimum width = 2.5em, minimum height = 1.75 em, right = 0.75 em of sum1] (Cont1) {\small$G_r$};
    	\node [sum, right = 3.5 em of Cont1] (sum3) {};
    	\node[draw = white] at (sum3.center) {$+$};
    	\node [smallblock, minimum width = 2.5em, minimum height = 1.75 em, right = 0.75 em of sum3] (Cont2) {\small$G_v$};
    	\node [right = 1.5 em of Cont2] (output) {};

    	\draw [->] (Ref_traj) node [above, xshift=0.25em]
    	{\scriptsize$\vect r_{c/w} \resolvedin{E} {\rm sp}$}-- (sum1);
    	\draw [->] ([yshift = -2em]sum1.south) -- node[xshift = 0.1em, yshift = -0.75em]{\scriptsize$\vect r_{c/w} \resolvedin{E} {\rm meas}$}(sum1.south);
    	\draw [->] (sum1) -- node [xshift=-1.4em, yshift = -1.5em]{$-$} (Cont1);
    	\draw [->] ([yshift = 2.5em]Cont1.north) -- node [xshift = -4.5em, yshift = 1.25em]{\scriptsize$\framedotE{\vect r}_{c/w}\resolvedin{E}{\rm sp,ff}$}(Cont1.north);
    	\draw [->] (Cont1) -- node [yshift = -0.1em]{\scriptsize$\framedotE{\vect r}_{c/w}\resolvedin{E} {\rm sp}$}(sum3);
    	\draw [->] (sum3) -- node [xshift=-1.4em, yshift = -1.42em]{$-$} (Cont2);
    	\draw [->] (Cont2) -- node [above,xshift = 0.5em, yshift = -0.1em] {\scriptsize$\vect f_\rmc\resolvedin{E}{\rm sp}$} (output);
    	\draw [->] ([yshift = -2em]sum3.south) -- node[xshift = 0.1em, yshift = -0.5em]{\scriptsize$\framedotE{\vect r}_{c/w}\resolvedin{E} {\rm meas}$}(sum3.south);

    \end{tikzpicture}
    }
    \caption{\footnotesize PX4 autopilot position controller. }
    \label{fig:PX4_autopilot_outer_loop}
\end{figure}

The Attitude controller consists of a static map f2q and two cascaded controllers $G_q$ and $G_\omega$ as shown in Figure \ref{fig:PX4_autopilot_inner_loop}. 
The static map f2q converts the force setpoint into the quaternion setpoint as described below.

\begin{figure}[h]
    \centering
    \resizebox{0.95\columnwidth}{!}{
    \begin{tikzpicture}[auto, node distance=2cm,>=latex']
        \node (ForceVector) {};
    	\node [smallblock, right = 1.5 em of ForceVector,  minimum width = 2.5em, minimum height = 1.75 em] (Attitude) {\small f2q};
    	\node [sum, right = 2.25 em of Attitude] (sum1) {};
    	\node[draw = white] at (sum1.center) {$+$};
    	\node [smallblock, right = 0.75 em of sum1,  minimum width = 2.5em, minimum height = 1.75 em] (Cont1) {\small$G_q$};
    	\node [sum, right = 1 em of Cont1] (sum2) {};
    	\node[draw = white] at (sum2.center) {$+$};
    	\node [smallblock, right = 0.75 em of sum2,  minimum width = 2.5em, minimum height = 1.75 em] (Cont2) {\small$G_\omega$};
    	\node [right = 1.25 em of Cont2] (output) {};
    	
    	\draw[->] (ForceVector.center) node [xshift=0.6em, yshift = 0.8em] {\scriptsize$\vect f_\rmc\resolvedin{E}{\rm sp}$} -- (Attitude.west);
    	\draw[->] (Attitude.east) -- (sum1.west) node [xshift=-1em, yshift = 0.6em] {\scriptsize$q_{\rm Q/E}^{\rm sp}$};
    	\draw [->] ([yshift = -1.2em]sum1.south) -- node [xshift=1.2 em, yshift = -1.3 em] {\scriptsize$q_{\rm Q/E}^ {\rm meas}$} (sum1.south);
    	\draw [->] (sum1.east) -- node [xshift=-1.5 em, yshift = -1.5 em]{$-$} (Cont1.west);
    	\draw [->] (Cont1.east) -- (sum2.west);
    	\draw [->] ([yshift = -1.2em]sum2.south) -- node [xshift=2.25 em, yshift = -1.3 em] {\scriptsize$\vect \omega_{\rm Q/E} \resolvedin{Q}{\rm meas}$} (sum2.south);
    	\draw [->] (sum2.east) -- node [xshift=-1.5 em, yshift = -1.5 em]{$-$} (Cont2.west);
    	\draw [->] ([yshift = 0.75em]Cont1.north) -- node [xshift=-1.5em, yshift = 0.9em] {\scriptsize$\dot \Psi_{\rm sp,ff}$} (Cont1.north);
    	\draw [->] (Cont1.east) -- ([xshift=0.35em]Cont1.east) |- node [xshift=1.8 em, yshift = -0.25em] {\scriptsize$\vect \omega_{\rm Q/E} \resolvedin{Q}{sp}$} ([yshift = 0.75em]Cont2.north) -- (Cont2.north);
    	\draw [->] (Cont2.east) -- (output.center) node [xshift=0.25 em, yshift = 0.6 em] {\scriptsize$\vect M_{\SQ/\rmc}\resolvedin{Q}{\rm sp}$};
    \end{tikzpicture}
    }
    \caption{\footnotesize PX4 autopilot attitude controller.}
    \label{fig:PX4_autopilot_inner_loop}
\end{figure}
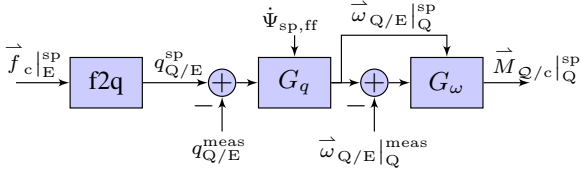

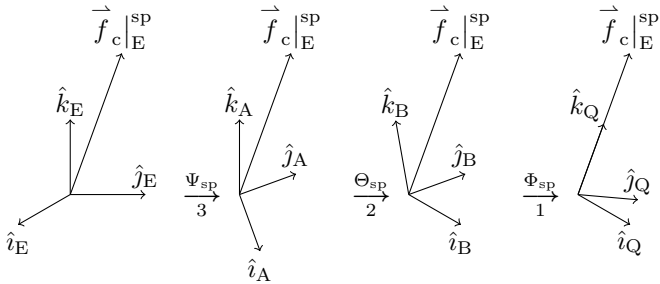
\begin{figure}[H]
    \centering
    \begin{tikzpicture}[auto, node distance=2cm]
        
        \node (O_E) at (0,0) {};
        \node (O_A) at (2.25,0) {};
        \node (O_B) at (4.50,0) {};
        \node (O_Q) at (6.75,0) {};
        
        \node (R1) at (1.75,0) {$\rotation{\Psi_{\rm sp}}{3}$};
        \node (R2) at (4,0) {$\rotation{\Theta_{\rm sp}}{2}$};
        \node (R3) at (6.25,0) {$\rotation{\Phi_{\rm sp}}{1}$};

        \draw[->] (O_E.center)--+(1,0) node [xshift=0, yshift = 7] {$\jhat_\rmE$} ;
        \draw[->] (O_E.center)--+(0,1) node [xshift=0, yshift = 7] {$\khat_\rmE$} ;
        \draw[->] (O_E.center)--+({.8*cos(210)},{.8*sin(210)}) node [xshift=0, yshift = -7] {$\ihat_\rmE$} ;
        
        \draw[->] (O_A.center)--+({.8*cos(20)},{.8*sin(20)}) node [xshift=0, yshift = 7] {$\jhat_\rmA$} ;
        \draw[->] (O_A.center)--+(0,1) node [xshift=0, yshift = 7] {$\khat_\rmA$} ;
        \draw[->] (O_A.center)--+({.8*cos(290)},{.8*sin(290)}) node [xshift=0, yshift = -7] {$\ihat_\rmA$} ;
        
        \draw[->] (O_B.center)--+({.8*cos(20)},{.8*sin(20)}) node [xshift=0, yshift = 7] {$\jhat_\rmB$} ;
        \draw[->] (O_B.center)--+({1*cos(100)},{1*sin(100)}) node [xshift=0, yshift = 7] {$\khat_\rmB$} ;
        \draw[->] (O_B.center)--+({.8*cos(330)},{.8*sin(330)}) node [xshift=0, yshift = -7] {$\ihat_\rmB$} ;
        
        \draw[->] (O_Q.center)--+({.8*cos(-5)},{.8*sin(-5)}) node [xshift=0, yshift = 7] {$\jhat_\rmQ$} ;
        \draw[->] (O_Q.center)--+({1*cos(70)},{1*sin(70)}) node [xshift=-7, yshift = 7] {$\khat_\rmQ$} ;
        \draw[->] (O_Q.center)--+({.8*cos(330)},{.8*sin(330)}) node [xshift=0, yshift = -7] {$\ihat_\rmQ$} ;
        
        \draw[->] (O_Q.center)--+({2*cos(70)},{2*sin(70)}) node [xshift=0, yshift = 10] {$\vect f_\rmc\resolvedin{E}{\rm sp}$};
        
        \draw[->] (O_E.center)--+({2*cos(70)},{2*sin(70)}) node [xshift=0, yshift = 10] {$\vect f_\rmc\resolvedin{E}{\rm sp}$};
        \draw[->] (O_A.center)--+({2*cos(70)},{2*sin(70)}) node [xshift=0, yshift = 10] {$\vect f_\rmc\resolvedin{E}{\rm sp}$};
        \draw[->] (O_B.center)--+({2*cos(70)},{2*sin(70)}) node [xshift=0, yshift = 10] {$\vect f_\rmc\resolvedin{E}{\rm sp}$};        
        
    \end{tikzpicture}

    \caption{\footnotesize 3-2-1 Euler angles that uniquely specify the attitude setpoint given the force setpoint and the azimuth setpoint. }
    \label{fig:force_to_EA}
\end{figure}
The force setpoint $\vect f_\rmc\resolvedin{E}{\rm sp}$ and the azimuth setpoint $\Psi_{\rm sp}$ uniquely determine the attitude setpoint.
%
Using the 3-2-1 Euler angle sequence shown in Figure \ref{fig:force_to_EA} and the force setpoint $\vect f_\rmc\resolvedin{E}{\rm sp}$, 
it follows that $\hat k_{\rmQ} \resolvedin{E}{sp}$ satisfies
\begin{align}
    \hat k_{\rmQ} \resolvedin{E}{sp} 
        =
            \dfrac
                {\vect f_\rmc\resolvedin{E}{\rm sp}}
                {\| \vect f_\rmc\resolvedin{E}{\rm sp} \|_2}.
\end{align}
%
Using the azimuth setpoint $\Psi_{\rm sp}$ specified by the mission planner and the rotation about $\khat_{\rmE}$ , it follows that 
\begin{align}
    \khat_\rmQ \resolvedin{A}{sp}
        =
            \SO_{\rm A/E} 
            \khat_\rmQ \resolvedin{E}{sp}
        =
            \SO_3(\Psi_{\rm sp})
            \khat_\rmQ \resolvedin{E}{sp}.
\end{align}
Since the frame $\rm F_B$ is obtained by rotating it about the $\jhat_{\rmA}$ axis so that $\vect f_\rmc\resolvedin{E}{\rm sp}$ lies in the $\jhat_\rmB-\khat_\rmB$ plane, it follows that
\begin{align}
    \tan \Theta_{\rm sp}
        =
            \dfrac
                {e_1^\rmT \khat_\rmQ \resolvedin{A}{sp}}
                {e_3^\rmT \khat_\rmQ \resolvedin{A}{sp}},
    \label{eq:Theta_sp}
\end{align}
and thus
\begin{align}
    \khat_\rmQ \resolvedin{B}{sp}
        =
            \SO_{\rm B/A} 
            \khat_\rmQ \resolvedin{A}{sp}
        =
            \SO_2(\Theta_{\rm sp})
            \khat_\rmQ \resolvedin{E}{sp}.
\end{align}
Finally, frame $\rm F_Q$ is obtained by rotating it about the $\ihat_{\rmB}$ axis so that $\khat_\rmQ \resolvedin{E}{sp}$ is along $\vect f_\rmc\resolvedin{E}{\rm sp}$.
It thus follows that
\begin{align}
    \tan \Phi_{\rm sp} 
        =
            \dfrac
                {e_2^\rmT \khat_\rmQ \resolvedin{B}{sp}}
                {e_3^\rmT \khat_\rmQ \resolvedin{B}{sp}}.
    \label{eq:Phi_sp}
\end{align}

%
Next, the attitude setpoint given by the 3-2-1 Euler angles 
$\Psi_{\rm sp},$
$\Theta_{\rm sp},$ and
$\Phi_{\rm sp}$ is converted to the quaternion form.
Note that the quaternion $q_{\rm Q/E}$ corresponding to 3-2-1 Euler angles $\Psi, \Theta,$ and $\Phi$ is given by
\begin{align}
    &q_{\rm Q/E} (\Psi, \Theta, \Phi)
        =
        \nn \\
        &
            \matl{r}
                \vspace{.5em}
                \cos \dfrac{\Phi}{2} \cos \dfrac{\Theta}{2} \cos \dfrac{\Psi}{2} +
                \sin \dfrac{\Phi}{2} \sin \dfrac{\Theta}{2} \sin \dfrac{\Psi}{2} 
                \\
                \vspace{.5em}
                -\cos \dfrac{\Phi}{2} \sin \dfrac{\Theta}{2} \sin \dfrac{\Psi}{2} +
                \sin \dfrac{\Phi}{2} \cos \dfrac{\Theta}{2} \cos \dfrac{\Psi}{2} 
                \\
                \vspace{.5em}
                \cos \dfrac{\Phi}{2} \sin \dfrac{\Theta}{2} \cos \dfrac{\Psi}{2} +
                \sin \dfrac{\Phi}{2} \cos \dfrac{\Theta}{2} \sin \dfrac{\Psi}{2} 
                \\
                \vspace{.5em}
                \cos \dfrac{\Phi}{2} \cos \dfrac{\Theta}{2} \sin \dfrac{\Psi}{2} -
                \sin \dfrac{\Phi}{2} \sin \dfrac{\Theta}{2} \cos \dfrac{\Psi}{2} 
            \matr
\end{align}
Using
the measured attitude $q_{\rm Q/E}^{\rm meas}$ and 
the setpoint attitude $q_{\rm Q/E}^ {\rm  sp}$, the attitude error $\tilde q $ is given by
\begin{align}
    \tilde q 
        \isdef
            (q_{\rm Q/E}^{\rm meas}) ^{-1}
            q_{\rm Q/E}^{{\rm  sp}} 
\end{align}
%
Finally, writing the quaternion $q = [\eta \ \varepsilon^\rmT],$ 
where $\eta \in [-1,1]$ and 
$\varepsilon$ is a $3 \times 1$ unit vector, 
the body-fixed angular velocity setpoint is given by
\begin{align}
    \vect \omega_{\rm Q/E} \resolvedin{Q}{sp}
        &=
            \frac{2}{\tau}
            {\rm sgn }( \hat \eta)  \hat \varepsilon,
            ,
        \label{eq:omega_controller_full}
\end{align}
where $\tau>0$ is a tuning parameter. 
Note that the attitude controller \eqref{eq:omega_controller_full} is an almost globally stabilizing controller \cite{Chaturvedi2011}.
However, the controller \eqref{eq:omega_controller_full} may not provide good tracking since the azimuth response is usually slower than the elevation and bank responses due to the larger moment of inertia about the azimuth axis. 
Alternatively, a mixed attitude controller consisting of a reduced attitude error and a feedforward azimuth-rate controller can be used to generate the body-fixed angular velocity setpoint as described below.

The {\it reduced attitude error} $\tilde q_{\rm red}$ is defined as the rotation that aligns $\khat_{\rmQ}$ with $\khat_{\rmQ}^{\rm sp}$  using the smallest angle of rotation, which is shown in Figure \ref{fig:red_quat_rotation}.
The reduced attitude error is thus given by
\begin{align}
    \tilde q_{\rm red}
        =
            \matl{c}
                \cos \dfrac{\alpha}{2} \\
                \sin \dfrac{\alpha}{2} 
                (\khat_\rmQ^{\rm meas} \times \khat_\rmQ^{\rm sp})
            \matr
    \label{eq:red_att_error}
\end{align}
where 
\begin{align}
    \alpha 
        \isdef
            \arccos
            (\khat_{\rmQ}^{\rm meas} \cdot \khat_{\rmQ}^{\rm sp})
\end{align}
and 
$\khat_\rmQ^{\rm meas} \times \khat_\rmQ^{\rm sp}$ is the axis of rotation.
\begin{figure}[H]
    \centering
    \begin{tikzpicture}[auto, node distance=2cm]
        
        \node (O_E) at (0,0) {};

        \draw[->] (O_E.center)--+({1*cos(110)},{1*sin(110)}) node [xshift=0, yshift = 7] {$\khat_\rmQ^{\rm sp}$} ;
        \draw[->] (O_E.center)--+({1*cos(40)},{1*sin(40)}) node [xshift=10, yshift = 7] {$\khat_\rmQ^{\rm meas}$} ;
        \draw[->] (O_E.center)--+({1*cos(320)},{1*sin(320)}) node [xshift=0, yshift = -7] 
        {$\khat_\rmQ^{\rm meas} \times \khat_\rmQ^{\rm sp}$} ;
        
        \draw [line width=1, -] ({.5*cos(110)},{.5*sin(110)}) arc (140:40:.5)
            node[xshift=-10, yshift=10] {$\alpha$}; 
        
    \end{tikzpicture}

    \caption{\footnotesize Reduced attitude error. The reduced attitude is the smallest rotation that aligns $\khat_{\rmQ}$ with $\khat_{\rmQ}^{\rm sp}$. }
    \label{fig:red_quat_rotation}
\end{figure}
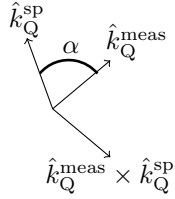
Using the reduced attitude error \eqref{eq:red_att_error}, the body-fixed angular-velocity setpoint is given by
\begin{align}
    \vect \omega_{\rm Q/E} \resolvedin{Q}{sp}
        &=
            K_q
            z_q
            +
            \dot \Psi_{\rm sp, ff} \khat_{\rmE} \resolvedin{Q}{} \nn \\
        &=
            K_q
            z_q
            +
            \dot \Psi_{\rm sp, ff} \SO_{\rm Q/E} e_3
            ,
        \label{eq:omega_controller_red}
\end{align}
where 
$z_q \isdef {\rm sgn }( \hat \eta_{\rm red})  \hat \varepsilon_{\rm red}, $
$K_q$ is a $3\times 3$ diagonal matrix and 
$\dot \Psi_{\rm sp, ff}$ is the feedworward azimuth rate setpoint specified by the mission planner.
Note that the diagonal entries of $K_q$ are the tuning gains.

Finally, the moment setpoint is given by 
\begin{align}
    \vect M_{\SQ/\rmc}\resolvedin{Q}{\rm sp}
        &=
            G_{\omega} 
            \matl{c}
                z_\omega \\
                \vect \omega_{\rm Q/E} \resolvedin{Q}{\rm sp}
            \matr,
\end{align}
where 
$z_\omega \isdef 
            \vect \omega_{\rm Q/E} \resolvedin{Q}{\rm sp}
            -
            \vect \omega_{\rm Q/E} \resolvedin{Q}{\rm meas}
            ,$
\begin{align}
    G_\omega 
        = 
        \matl{cc}
            K_{\omega,\rmP} + 
            \dfrac{K_{\omega,\rmI}}{\shiftq-1} + 
            K_{\omega,\rmD} \dfrac{\shiftq- 1}{\shiftq}
            & 
            K_{\omega, \rm ff}
        \matr,
\end{align}            
and 
$K_{\omega,\rmP}, K_{\omega,\rmI}, K_{\omega,\rmD},$ and  $K_{\omega,\rm ff}$ are $3\times3$ diagonal matrices. 
Note that diagonal entries of $K_{\omega,\rmP}, K_{\omega,\rmI},K_{\omega,\rmD},$ and $K_{\omega,\rm ff}$ are the tuning gains.
%



The control system implemented in the PX4 autopilot thus consists of 27 gains.
In particular, the position controller includes three gains in $G_r$ and nine gains in $G_v$; and
the attitude controller includes three gains in $G_q$ and 12 gains in $G_\omega.$
In practice, these 27 gains are manually tuned and require considerable expertise.

\section{Adaptive Digital Control Algorithm}
\label{sec:PID_Algo}

This section describes the retrospective cost adaptive control (RCAC) technique that is used to update the control law in a sampled-data feedback loop.
RCAC is described in detail in \cite{rahmanCSM2017} and its extension to digital PID control is given in \cite{rezaPID}.
Consider the control law
\begin{align}
    u_k 
        =
            \phi_k \theta_k,
    \label{eq:uk_reg}
\end{align}
where, for all $k\ge0$, 
the regressor $\phi_k \in \BBR^{l_u \times l_\theta}$ contains the measurements and $l_\theta$ depends on the structure of the controller.
The controller coefficients $\theta_k \in \BBR^{l_\theta}$ are optimized by RCAC as described below. 

Consider the SISO PID controller with a feedforward term
\begin{align}
    u_k
        =
            K_{\rmp,k} z_{k-1} +
            K_{\rmi,k} \gamma_{k-1} +
            K_{\rmd,k} (z_{k-1} - z_{k-2}) +
            K_{{\rm ff},k} r_k
            ,
    \label{eq:uk_PID}
\end{align}
where $K_{\rmp,k}, K_{\rmi,k}, K_{\rmd,k}, $ and $K_{{\rm ff},k}$ are time-varying gains to be optimized, 
$z_k$ is an error variable,
$r_k$ is the feedforward signal, 
and, for all $k\ge0$,
\begin{align}
    \gamma_k 
        \isdef
            \sum_{i=0}^{k} z_{i}.
\end{align}
Note that the integrator state is computed recursively using $\gamma_k = \gamma_{k-1} + z_{k}$.
For all $k\ge0$, 
the regressor $\phi_k$ 
and 
the controller coefficient $\theta_k$ 
in \eqref{eq:uk_PID} are given by
\begin{align}
    \phi_k
        \isdef
            \matl{c}
                z_{k-1} \\
                \gamma_{k-1} \\
                z_{k-1} - z_{k-2} \\
                r_k
            \matr^\rmT, \quad
    \theta_k
        \isdef
            \matl{c}
                K_{\rmp,k} \\
                K_{\rmi,k} \\
                K_{\rmd,k} \\
                K_{{\rm ff},k}
            \matr 
            \in \BBR^{4}.
    \label{eq:phi_theta_def}
\end{align}
Note that various MIMO controller parameterizations are shown in \cite{goel_2020_sparse_para}.

To determine the controller gains $\theta_k$, let $\theta \in \BBR^{l_\theta}$, and consider the \textit{retrospective performance variable} defined by
\begin{align}
    \hat{z}_{k}(\theta)
        \isdef
            z_k + 
            \sigma (\phi_{k-1} \theta - u_{k-1}),
    \label{eq:zhat_def}
\end{align}
where $\sigma$ is either $1$ or $-1$ depending on whether the sign of the leading numerator coefficient of the transfer function from $u_k$ to $z_k$ is positive or negative, respectively.
Furthermore, define the \textit{retrospective cost function} $J_k \colon \BBR^{l_\theta} \to [0,\infty)$ by
\begin{align}
    J_k(\theta) 
        \isdef
            \sum_{i=0}^k
                \hat{z}_{k}(\theta)^2 +
                (\theta-\theta_0)^\rmT 
                P_0^{-1}
                (\theta-\theta_0),
    \label{eq:RetCost_def}
\end{align}
where $\theta_0\in\BBR^{l_\theta}$ is the initial vector of PID gains and $P_0\in\BBR^{l_\theta\times l_\theta}$ is positive definite.
%

\begin{proposition}
    Consider \eqref{eq:uk_reg}--\eqref{eq:RetCost_def}, 
    where $\theta_0 \in \BBR^{l_\theta}$ and $P_0 \in \BBR^{l_\theta \times l_\theta}$ is positive definite. 
    Furthermore, for all $k\ge0$, denote the minimizer of $J_k$ given by \eqref{eq:RetCost_def} by
    \begin{align}
        \theta_{k+1}
            \isdef
                \underset{ \theta \in \BBR^n  }{\operatorname{argmin}} \
                J_k({\theta}).
        \label{eq:theta_minimizer_def}
    \end{align}
    Then, for all $k\ge0$, $\theta_{k+1}$ is given by 
    \begin{align}
        \theta_{k+1} 
            &=
                \theta_k  + 
                 P_{k+1}\phi_{k-1}^\rmT
                 [ z_k + \sigma(\phi_{k-1} \theta_k - u_{k-1}) ]
                 , \label{eq:theta_update}
    \end{align}
    where 
    \begin{align}
        P_{k+1} 
            &=
                P_{k}
                -  \frac
                    { P_{k}\phi_{k-1}^\rmT  \phi_{k-1} P_{k} }
                    { 1 +   \phi_{k-1} P_{k} \phi_{k-1}^\rmT  }.
        \label{eq:P_update_noInverse}
    \end{align}
\end{proposition}
\begin{proof}
See \cite{AseemRLS}
\end{proof}

\section{Adaptive PX4 Autopilot}
\label{sec:adaptiveAugmentation}
This section describes the augmentation of the PX4 autopilot with adaptive controllers.
The adaptive digital autopilot is constructed by modifying the PX4 autopilot.
%
The controllers $G_r$ and $G_v$ in the position controller are augmented with adaptive components as shown in Figure \ref{fig:Augmented_PX4_autopilot_outer_loop}.
The velocity setpoint is given by
\begin{align}
    \framedotE{\vect r}_{c/w}\resolvedin{E}{\rm sp}
        &=
            K_r
            z_r
            +
            \framedotE{\vect r}_{c/w}\resolvedin{E}{\rm sp,ff}
            + u_r
             ,
\end{align}
where 
$u_r = \phi_r \theta_r$ is the output of the adaptive $G_r,$ 
$\phi_r \isdef {\rm diag} (z_r),$ and 
$\theta_r \in \BBR^3$ is updated using \eqref{eq:theta_update}, \eqref{eq:P_update_noInverse}.
Note that the structure of the adaptive $G_r$ is same as that of $G_r.$
The timestep $k$ in the notation for $u,$ $\phi,$ and $\theta$ is omitted to improve readability.  

Next, the force setpoint $\vect f_\rmc \resolvedin{E}{{\rm  sp}}$ is given by
\begin{align}
    \vect f_\rmc \resolvedin{E}{{\rm  sp}}
        =
            G_v
            z_v
            + u_v,
\end{align}
where 
$u_v = \phi_v \theta_v$ is the output of the adaptive $G_v,$ 
\begin{align}
    \phi_v
        = 
            \matl{ccc}
                \phi_{1,v} & 0 & 0 \\
                0 & \phi_{2,v} & 0 \\
                0 & 0 & \phi_{3,v}
            \matr,
\end{align}
$\theta_v \in \BBR^9$ is updated using \eqref{eq:theta_update}, \eqref{eq:P_update_noInverse}, and, for $i \in \{1,2,3\}$,
\begin{align}
    \phi_{i,v}
        \isdef 
            \matl{ccc}
                z_{i,v,k-1} &
                \gamma_{i,v,k-1} &
                z_{i,v,k-1} - z_{i,v,k-2} 
            \matr,
\end{align}
where $z_{i,v,k}$ is the $i$th component of the vector $z_v$ at the timestep $k.$

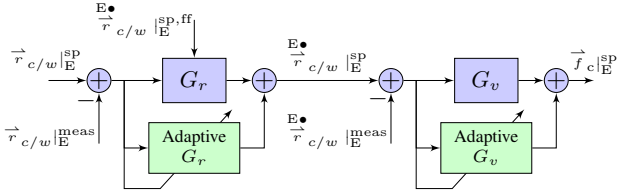
\begin{figure}[H]
    \centering
    \resizebox{\columnwidth}{!}{
    \begin{tikzpicture}[auto, node distance=2cm,>=latex']
        
        \node (Ref_traj) {};
    	\node [sum, right = 1.5 em of Ref_traj] (sum1) {};
    	\node[draw = white] at (sum1.center) {$+$};
    	\node [smallblock, minimum width = 2.5em, minimum height = 1.75 em,right = 2 em of sum1] (Cont1) {\small$G_r$};
    	\node [sum, right = 1 em of Cont1] (sum2) {};
    	\node[draw = white] at (sum2.center) {$+$};
    	\node [sum, right = 6 em of Cont1] (sum3) {};
    	\node[draw = white] at (sum3.center) {$+$};
    	\node [smallblock, minimum width = 2.5em, minimum height = 1.75 em, right = 2em of sum3] (Cont2) {\small$G_v$};
    	\node [sum, right = 1 em of Cont2] (sum4) {};
    	\node[draw = white] at (sum4.center) {$+$};
    	\node [right = 1 em of sum4] (output) {};

    	\draw [->] (Ref_traj) node [above, xshift=0.25em]
    	{\tiny$\vect r_{c/w} \resolvedin{E} {\rm sp}$}-- (sum1);
    	\draw [->] ([yshift = -2em]sum1.south) -- node[xshift = 0.25em, yshift = -0.75em]{\tiny$\vect r_{c/w} \resolvedin{E} {\rm meas}$}(sum1.south);
    	\draw [->] (sum1) -- node [xshift=-2em, yshift = -1.5em]{$-$} (Cont1);
    	\draw [->] ([yshift = 1.5em]Cont1.north) -- node [xshift = -4.25em, yshift = 0.75em]{\tiny$\framedotE{\vect r}_{c/w}\resolvedin{E}{\rm sp,ff}$}(Cont1.north);
    	\draw [->] (Cont1.east) -- (sum2.west);
    	\draw [->] (sum2.east) -- node [yshift = -0.1em]{\tiny$\framedotE{\vect r}_{c/w}\resolvedin{E} {\rm sp}$}(sum3);
    	\draw [->] (sum3) -- node [xshift=-2em, yshift = -1.42em]{$-$} (Cont2);
    	\draw [->] (Cont2.east) -- (sum4.west);
    	\draw [->] (sum4.east) -- node [above,xshift = 0.5em, yshift = -0.1em] {\tiny$\vect f_\rmc\resolvedin{E}{\rm sp}$} (output);
    	\draw [->] ([yshift = -2em]sum3.south) -- node[xshift = 0.25em, yshift = -0.5em]{\tiny$\framedotE{\vect r}_{c/w}\resolvedin{E} {\rm meas}$}(sum3.south);
    	
    	\node [right = 0.15 em of sum1.east] (sum1_hook) {};
    	\draw [->] (sum1_hook.center) -- +(0,-1.5) -- +(.50,-1.5) -- +(1.5,-.4);
    	\node [smallblock, fill=green!20, minimum width = 0.5em, minimum height = 2 em, inner sep=0.25pt, below = 0.75 em of Cont1] (Cont1_adp) {\scriptsize $\begin{array}{c} \text{Adaptive} \\ G_r\end{array}$};
    	\draw [->] (sum1_hook.center) |- (Cont1_adp.west);
    	\draw [->] (Cont1_adp.east) -| (sum2.south);
    	
    	\node [right = 0.15 em of sum3.east] (sum3_hook) {};
    	\draw [->] (sum3_hook.center) -- +(0,-1.5) -- +(.50,-1.5) -- +(1.5,-.4);
    	\node [smallblock, fill=green!20, minimum width = 0.5em, minimum height = 2 em, inner sep=0.25pt, below = 0.75 em of Cont2] (Cont2_adp) {\scriptsize $\begin{array}{c} \text{Adaptive} \\ G_v\end{array}$};
    	\draw [->] (sum3_hook.center) |- (Cont2_adp.west);
    	\draw [->] (Cont2_adp.east) -| (sum4.south);

    \end{tikzpicture}
    }
    \caption{\footnotesize Adaptive PX4 autopilot position controller. }
    \label{fig:Augmented_PX4_autopilot_outer_loop}
\end{figure}

The controllers $G_q$ and $G_\omega$ in the attitude controller are augmented with adaptive components as shown in Figure \ref{fig:Augmented_PX4_autopilot_inner_loop}.
The body-fixed angular velocity setpoint is given by
\begin{align}
    \vect \omega_{\rm Q/E} \resolvedin{Q}{sp}
        &=
            K_q
            z_q
            +
            \dot \Psi_{\rm sp,ff} \SO_{\rm Q/E} e_3
            +
            u_q
            ,
\end{align}
where 
$u_q = \phi_q \theta_q$ is the output of the adaptive $G_q,$  
$\phi_q \isdef {\rm diag} (z_q),$ and 
$\theta_q \in \BBR^3$ is updated using \eqref{eq:theta_update}, \eqref{eq:P_update_noInverse}.

Finally, the moment setpoint is given by 
\begin{align}
    \vect M_{\SQ/\rmc}\resolvedin{Q}{\rm sp}
        =
            G_{\omega} 
            \matl{c}
                z_\omega \\
                \vect \omega_{\rm Q/E} \resolvedin{Q}{\rm sp}
            \matr
            +
            u_\omega,
\end{align}
where 
$u_\omega = \phi_\omega \theta_\omega$ is the output of the adaptive $G_\omega,$ 
\begin{align}
    \phi_\omega
        = 
            \matl{ccc}
                \phi_{1,\omega} & 0 & 0 \\
                0 & \phi_{2,\omega} & 0 \\
                0 & 0 & \phi_{3,\omega}
            \matr,
\end{align}
$\theta_\omega \in \BBR^{12}$ is updated using \eqref{eq:theta_update}, \eqref{eq:P_update_noInverse}, and, for $i \in \{1,2,3\}$,
\begin{align}
    \phi_{i,\omega}
        \isdef 
            \matl{c}
                z_{i,\omega,k-1} \\
                \gamma_{i,\omega,k-1} \\
                z_{i,\omega,k-1} - z_{i,\omega,k-2} \\
                e_i^\rmT \vect \omega_{\rm Q/E} \resolvedin{Q}{\rm sp}
            \matr^\rmT.
\end{align}

\begin{figure}[h]
    \centering
    \resizebox{\columnwidth}{!}{
    \begin{tikzpicture}[auto, node distance=2cm,>=latex']
        \node (ForceVector) {};
    	\node [smallblock, right = 1.5 em of ForceVector,  minimum width = 2.5em, minimum height = 1.75 em] (Attitude) {\small f2q};
    	\node [sum, right = 2.25 em of Attitude] (sum1) {};
    	\node[draw = white] at (sum1.center) {$+$};
    	\node [smallblock, right = 2.35 em of sum1,  minimum width = 2.5em, minimum height = 1.75 em] (Cont1) {\small$G_q$};
    	\node [sum, right = 1.2 em of Cont1] (sum25) {};
    	\node[draw = white] at (sum25.center) {$+$};
    	\node [sum, right = 4.5 em of Cont1] (sum2) {};
    	\node[draw = white] at (sum2.center) {$+$};
    	\node [smallblock, right = 3.2 em of sum2,  minimum width = 2.5em, minimum height = 1.75 em] (Cont2) {\small$G_\omega$};
    	\node [sum, right = 1.2 em of Cont2] (sum3) {};
    	\node[draw = white] at (sum3.center) {$+$};
    	\node [right = 1.25 em of sum3] (output) {};
    	
    	\draw[->] (ForceVector.center) node [xshift=0.6em, yshift = 0.8em] {\scriptsize$\vect f_\rmc\resolvedin{E}{\rm sp}$} -- (Attitude.west);
    	\draw[->] (Attitude.east) -- (sum1.west) node [xshift=-1em, yshift = 0.6em] {\scriptsize$q_{\rm Q/E}^{\rm sp}$};
    	\draw [->] ([yshift = -1.2em]sum1.south) -- node [xshift=1.2 em, yshift = -1.3 em] {\scriptsize$q_{\rm Q/E}^ {\rm meas}$} (sum1.south);
    	\draw [->] (sum1.east) -- node [xshift=-2.2 em, yshift = -1.5 em]{$-$} (Cont1.west);
    	\draw [->] (Cont1.east) -- (sum25.west);
    	\draw [->] (sum25.east) -- (sum2.west);
    	\draw [->] ([yshift = -1.2em]sum2.south) -- node [xshift=2.25 em, yshift = -1.3 em] {\scriptsize$\vect \omega_{\rm Q/E} \resolvedin{Q}{\rm meas}$} (sum2.south);
    	\draw [->] (sum2.east) -- node [xshift=-2.6 em, yshift = -1.5 em]{$-$} (Cont2.west);
    	\draw [->] ([yshift = 0.75em]Cont1.north) -- node [xshift=-1.5em, yshift = 0.9em] {\scriptsize$\dot \Psi_{\rm sp,ff}$} (Cont1.north);
    	\draw [->] (sum25.east) -- ([xshift=0.35em]sum25.east) |- node [xshift=2 em, yshift = -0.25em] {\scriptsize$\vect \omega_{\rm Q/E} \resolvedin{Q}{sp}$} ([yshift = 0.75em]Cont2.north) -- (Cont2.north);
    	\draw [->] (Cont2.east) -- (sum3.west);
    	\draw [->] (sum3.east) -- (output.center) node [xshift=0.25 em, yshift = 0.6 em] {\scriptsize$\vect M_{\SQ/\rmc}\resolvedin{Q}{\rm sp}$};
    	
    	\node [right = 0.5 em of sum1.east] (sum1_hook) {};
    	\draw [->] (sum1_hook.center) -- +(0,-1.5) -- +(.50,-1.5) -- +(1.5,-.4);
    	\node [smallblock, fill=green!20, minimum width = 0.5em, minimum height = 2 em, inner sep=0.25pt, below = 0.75 em of Cont1] (Cont1_adp) {\scriptsize $\begin{array}{c} \text{Adaptive} \\ G_q\end{array}$};
    	\draw [->] (sum1_hook.center) |- (Cont1_adp.west);
    	\draw [->] (Cont1_adp.east) -| (sum25.south);
    	
    	\node [right = 1.35 em of sum2.east] (sum2_hook) {};
    	\draw [->] (sum2_hook.center) -- +(0,-1.5) -- +(.50,-1.5) -- +(1.5,-.4);
    	\node [smallblock, fill=green!20, minimum width = 0.5em, minimum height = 2 em, inner sep=0.25pt, below = 0.75 em of Cont2] (Cont2_adp) {\scriptsize $\begin{array}{c} \text{Adaptive} \\ G_\omega\end{array}$};
    	\draw [->] (sum2_hook.center) |- (Cont2_adp.west);
    	\draw [->] (Cont2_adp.east) -| (sum3.south);
    	
    	\draw [->] (sum25.east) -- ([xshift=0.35em]sum25.east) |- ([xshift = 1em, yshift = -1em]Cont2_adp.south) -- ([xshift = 1em]Cont2_adp.south);

    \end{tikzpicture}
    }
    \caption{\footnotesize Adaptive PX4 autopilot attitude controller.}
    \label{fig:Augmented_PX4_autopilot_inner_loop}
\end{figure}
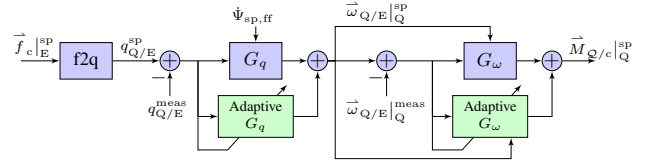

\section{Flight Tests}
\label{sec:flight_tests}
This section describes the results of the flight tests conducted with the adaptive PX4 autopilot.
The quadcopter is commanded to follow the trajectory shown in Figure \ref{fig:mair} in all flight tests.
First, the adaptive digital autopilot is tested with the SITL simulator. 
In this work, jMAVSim\footnote{\href{https://github.com/PX4/jMAVSim}{https://github.com/PX4/jMAVSim}} is used to simulate the quadcopter dynamics.
The default controller gains and the actuator constraints in PX4 are specified in the
\verb|mc_pos_control_param.c| and 
\verb|mc_att_control_param.c|\footnote{\href{https://github.com/ankgoel8188/Firmware}{https://github.com/ankgoel8188/Firmware}}.
%
%
Table \ref{tab:RCPE_variables_SITL} shows the hyperparameters used by RCAC in the simulated flight tests.
\begin{table}[h]
    \centering
    \begin{tabular}{|c|c|c|c|}
        \hline
        Variable & Controller coefficient  & $\sigma$  & $P_0$
        \\ \hline
        $u_r$ & $\theta_r$ & $I_3$ & $0.01$
        \\ \hline
        $u_v$ & $\theta_v$ & $I_9$ & $0.001$
        \\ \hline
        $u_q$ & $\theta_q$ & $I_3$ & $0.01$
        \\ \hline
        $u_\omega$ & $\theta_\omega$ & $I_{12}$ & $0.001$
        \\ \hline
    \end{tabular}
    \caption{\footnotesize Hyperparameters used by RCAC in the adaptive PX4 autopilot.}
    \label{tab:RCPE_variables_SITL}
\end{table}

%
\begin{figure}
    \centering
    \includegraphics[width=.4\textwidth]{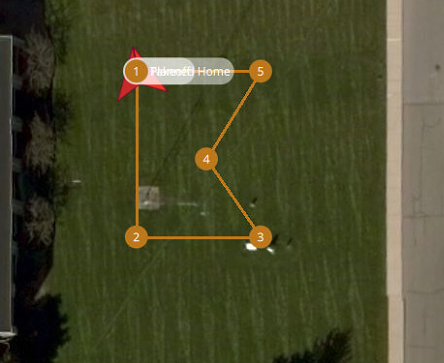}
    \caption{\footnotesize Mission plan for the flight tests conducted at the M-air at The University of Michigan, Ann Arbor. }
    \label{fig:mair}
\end{figure}

The default controller gains in PX4 are well-tuned for the jMAVSim simulator. 
To investigate the potential improvements in the performance, the default controller gains in PX4 are multiplied by a scalar $\alpha_{\rmp}$ to degrade the performance of the autopilot. 
This is equivalent to the case of poor initial choice of controller gains. 
The baseline performance is obtained by setting $\alpha_\rmP=1$ and the degraded autopilot performance by setting $\alpha_\rmP<1.$
The solid blue trace in Figure \ref{fig:ACC21_PX4_SITL_traj_comp} shows the trajectory-following response of the jMAVSim model in the baseline case. 
%
%
The solid red trace in Figure \ref{fig:ACC21_PX4_SITL_traj_comp} shows the trajectory-following response in the case where $\alpha_\rmp = 0.3.$
Note the large overshoots.
%
Next, the adaptive digital autopilot is used to fly the jMAVSim model. 
The dashed blue trace in Figure \ref{fig:ACC21_PX4_SITL_traj_comp} shows the trajectory-following response in the case where $\alpha_\rmp = 1,$ and 
the dashed red trace in Figure \ref{fig:ACC21_PX4_SITL_traj_comp} shows the trajectory-following response in the case where $\alpha_\rmp = 0.3.$
Note that with $\alpha_\rmp = 1,$ the performance is similar to the baseline case since the default controller is tuned well for the jMAVSim model.
However, in the case of $\alpha_\rmp = 0.3,$ the adaptive digital autopilot recovers the baseline performance. 
Furthermore, as shown by the dashed red trace in Figure \ref{fig:ACC21_PX4_SITL_traj_comp}, the trajectory-following response after the second waypoint is similar to the baseline case.
%

The corresponding azimuth errors for the four cases are shown in Figure \ref{fig:ACC21_PX4_SITL_RCAC_yaw_comp}.
Note that the mission takes about 65 s to complete compared to about 40 s in the baseline case. 
Figure \ref{fig:ACC21_PX4_SITL_RCAC_actuator_comp} shows the corresponding  thrust and the moment commands generated by the stock autopilot and the adaptive digital autopilot.
Note that, in the case where $\alpha_\rmp=1,$ the adaptive digital autopilot marginally improves the performance, and 
in the case where $\alpha_\rmp=0.3,$ the adaptive digital autopilot recovers the baseline performance.
Finally, Figure \ref{fig:ACC21_PX4_SITL_RCAC_theta} shows the controller gains optimized by RCAC in the adaptive digital autopilot.
Note that the magnitude of the adaptive gains increase as $\alpha_\rmp$ is reduced.  
This suggests that RCAC compensates for the poor choice of gains in the fixed-gain controllers.

\begin{figure}[h]
    \centering
    \includegraphics[width=.42\textwidth]{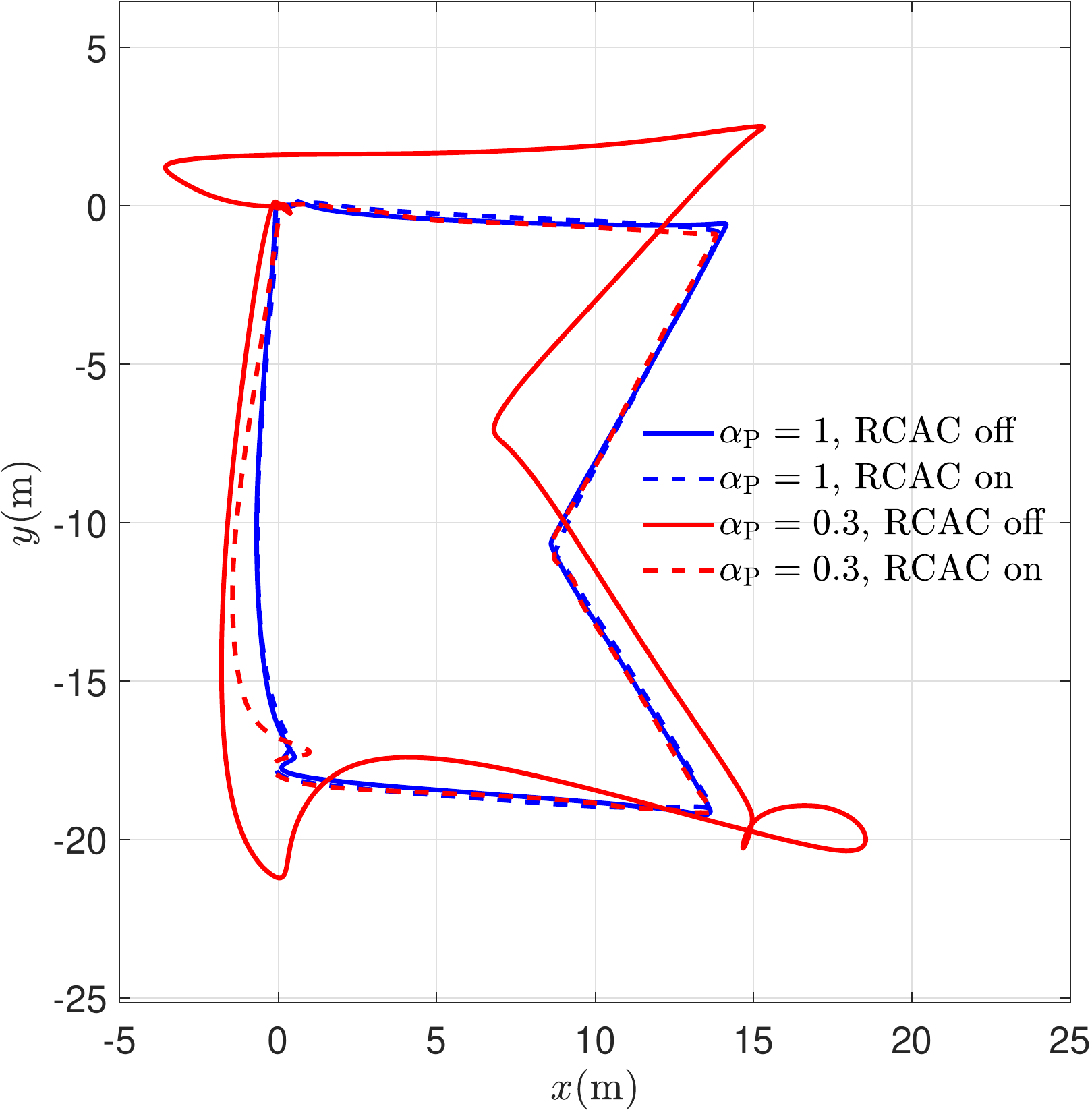}
    \caption{\footnotesize Closed-loop trajectory-following response of the jMAVSim model. 
    The four traces show the trajectory-following response with $\alpha_\rmp=1$ and $\alpha_\rmp=0.3$ with the stock PX4 and the adaptive PX4 autopilot. 
    }
    \label{fig:ACC21_PX4_SITL_traj_comp}
\end{figure}

\begin{figure}[h]
    \centering
    \includegraphics[width=.42\textwidth]{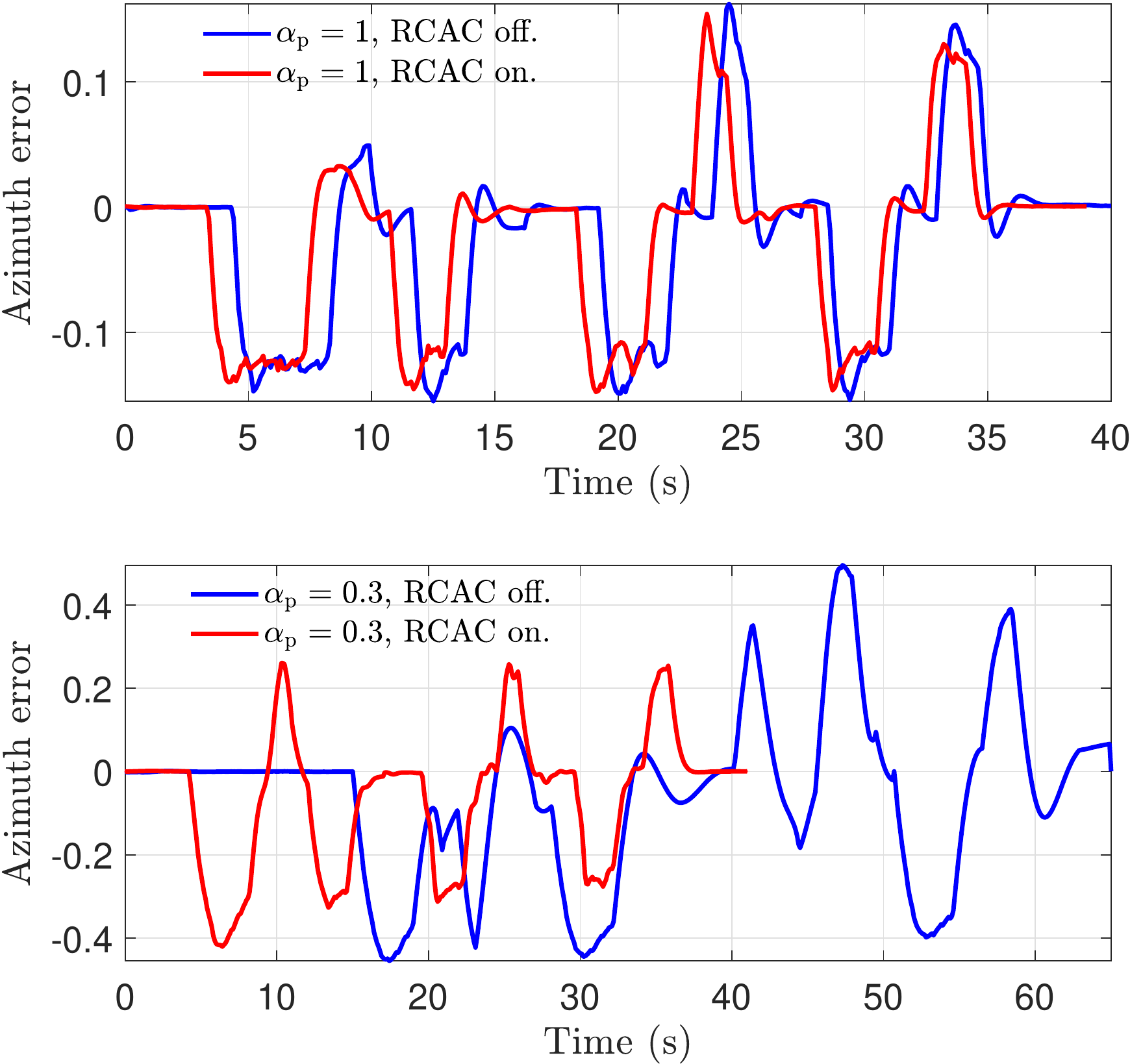}
    \caption{\footnotesize \footnotesize  Closed-loop azimuth response of the jMAVSim model. 
    The blue trace shows the azimuth error with the stock PX4 autopilot and the
    red trace shows the azimuth error with adaptive PX4 autopilot for two values of $\alpha_\rmp.$
    Note that, in the case of $\alpha_\rmp=1,$ the azimuth error is marginally better with RCAC, whereas in the case of $\alpha_\rmp=0.3,$ note that RCAC optimizes the controller to improve the response and the mission is completed in the same time taken in the baseline case. 
    }
    \label{fig:ACC21_PX4_SITL_RCAC_yaw_comp}
\end{figure}

\begin{figure}[h]
    \centering
    \includegraphics[width=.42\textwidth]{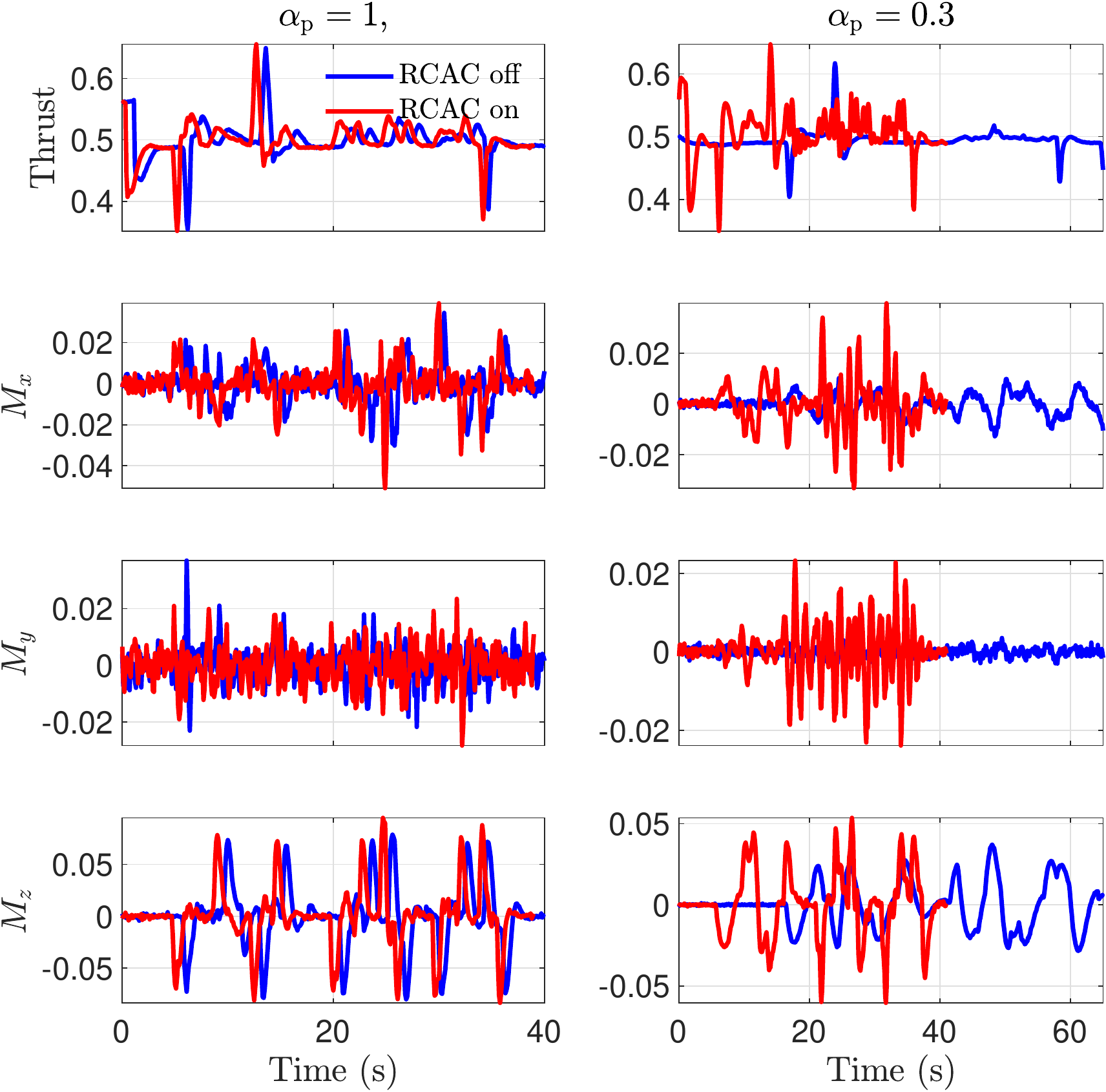}
    \caption{\footnotesize Thrust and the moment commands applied to the jMAVSim model. 
    }
    \label{fig:ACC21_PX4_SITL_RCAC_actuator_comp}
\end{figure}

\begin{figure}[h]
    \centering
    \includegraphics[width=.42\textwidth]{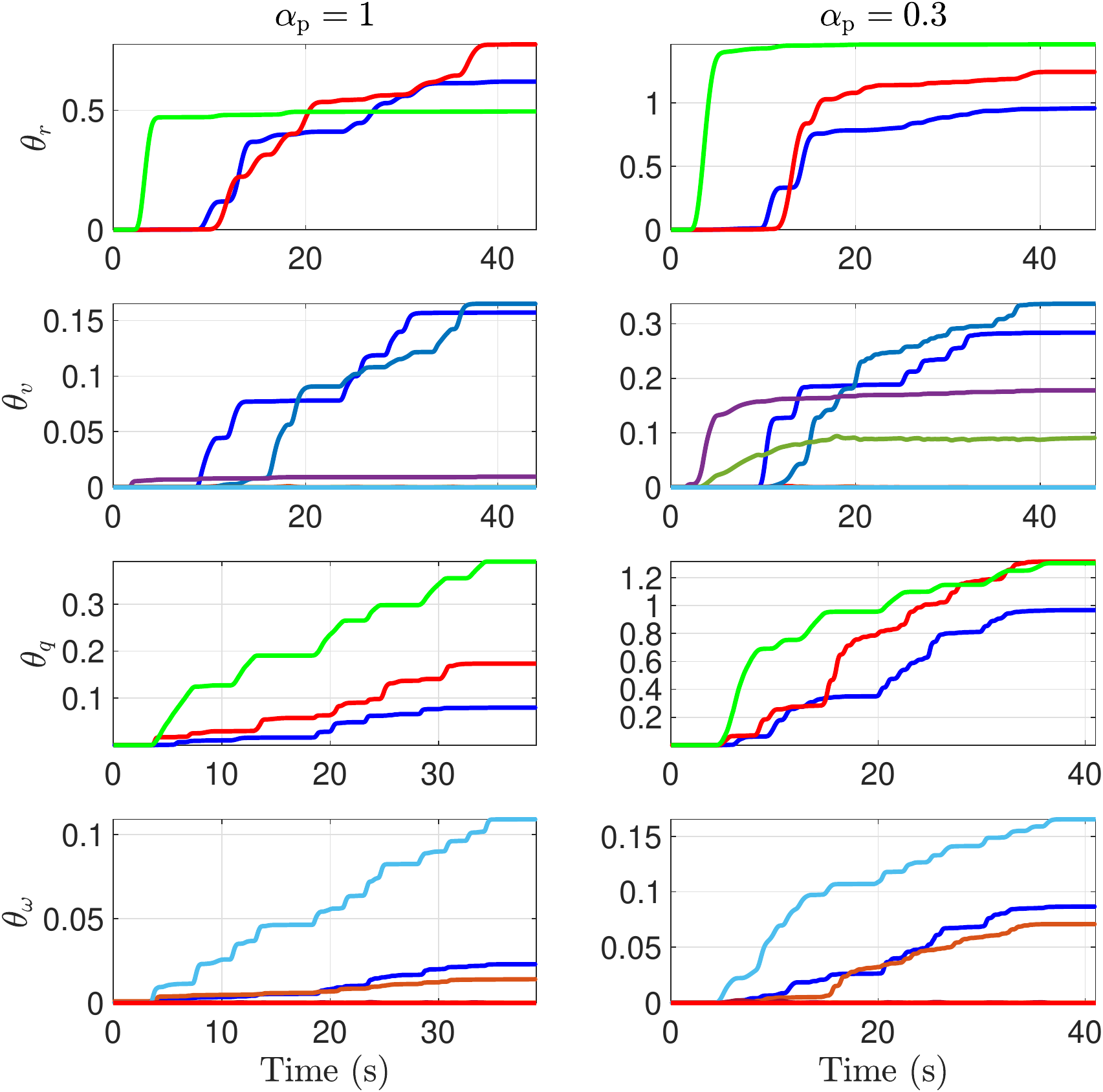}
    \caption{\footnotesize Adaptive gains optimized by RCAC in the adaptive digital autopilot. 
    Note that the magnitude of the adaptive gains increase as $\alpha_\rmp$ is reduced.  
    }
    \label{fig:ACC21_PX4_SITL_RCAC_theta}
\end{figure}

Next, the Holybro X500 quadcopter is flown in the M air facility at The University of Michigan, Ann Arbor. 
%
The default controller gains and the actuator constraints for the Holybro X500 quadcopter are specified in the QGroundControl mission planner. 
%
%
Table \ref{tab:RCPE_variables_HITL} shows the hyperparameters used by RCAC in the simulated flight tests.
It was observed that adaptive controller gain multiplying $z_{1,\omega}$ diverged and destabilized the flight. 
Thus, in all aubsequent flight test, it was set to zero. 

\begin{table}[h]
    \centering
    \begin{tabular}{|c|c|c|c|}
        \hline
        Variable & Controller coefficient  & $\sigma$  & $P_0$
        \\ \hline
        $u_r$ & $\theta_r$ & $I_3$ & $0.01$
        \\ \hline
        $u_v$ & $\theta_v$ & $I_9$ & $0.001$
        \\ \hline
        $u_q$ & $\theta_q$ & $I_3$ & $0.01$
        \\ \hline
        $u_\omega$ & $\theta_\omega$ & $I_{12}$ & $0.0001$
        \\ \hline
    \end{tabular}
    \caption{\footnotesize Hyperparameters used by RCAC in the adaptive PX4 autopilot.}
    \label{tab:RCPE_variables_HITL}
\end{table}

The solid blue trace in Figure \ref{fig:ACC21_PX4_HITL_traj_comp} shows the trajectory-following response of the Holybro X500 quadcopter in the baseline case. 
Next, the autopilot performance is degraded by setting $\alpha_\rmp = 0.5$ and $\alpha_\rmp=0.3.$ 
The solid red trace in Figure \ref{fig:ACC21_PX4_HITL_traj_comp} shows the trajectory-following response with $\alpha_\rmp = 0.5.$
However, with the stock PX4 autopilot, the quadcopter does not takeoff in the case where $\alpha_\rmp=0.3.$ 
%
Next, the adaptive digital autopilot is used to fly the Holybro X500 quadcopter. 
%
%
The dashed blue, red, and the green traces in Figure \ref{fig:ACC21_PX4_SITL_traj_comp} show the trajectory-following response with $\alpha_\rmp = 1,$  $\alpha_\rmp = 0.5,$ and $\alpha_\rmp = 0.3,$ respectively.
%
%

The corresponding azimuth errors for the five cases are shown in Figure \ref{fig:ACC21_PX4_HITL_RCAC_yaw_comp}.
Note that the mission takes about 70 s to complete with the degraded autopilot compared to about 55 s in the baseline case. 
Figure \ref{fig:ACC21_PX4_HITL_RCAC_actuator_comp} shows the corresponding  thrust and the moment commands generated by the stock autopilot and the adaptive digital autopilot in all five cases.
Finally, Figure \ref{fig:ACC21_PX4_HITL_RCAC_theta} shows the controller gains optimized by RCAC in the adaptive digital autopilot. 
Note that, in the case where $\alpha_\rmp=1,$ the adaptive digital autopilot marginally improves the performance, and, 
in the case where $\alpha_\rmp<1,$ the adaptive digital autopilot recovers the baseline performance.

\begin{figure}[h]
    \centering
    \includegraphics[width=.42\textwidth]{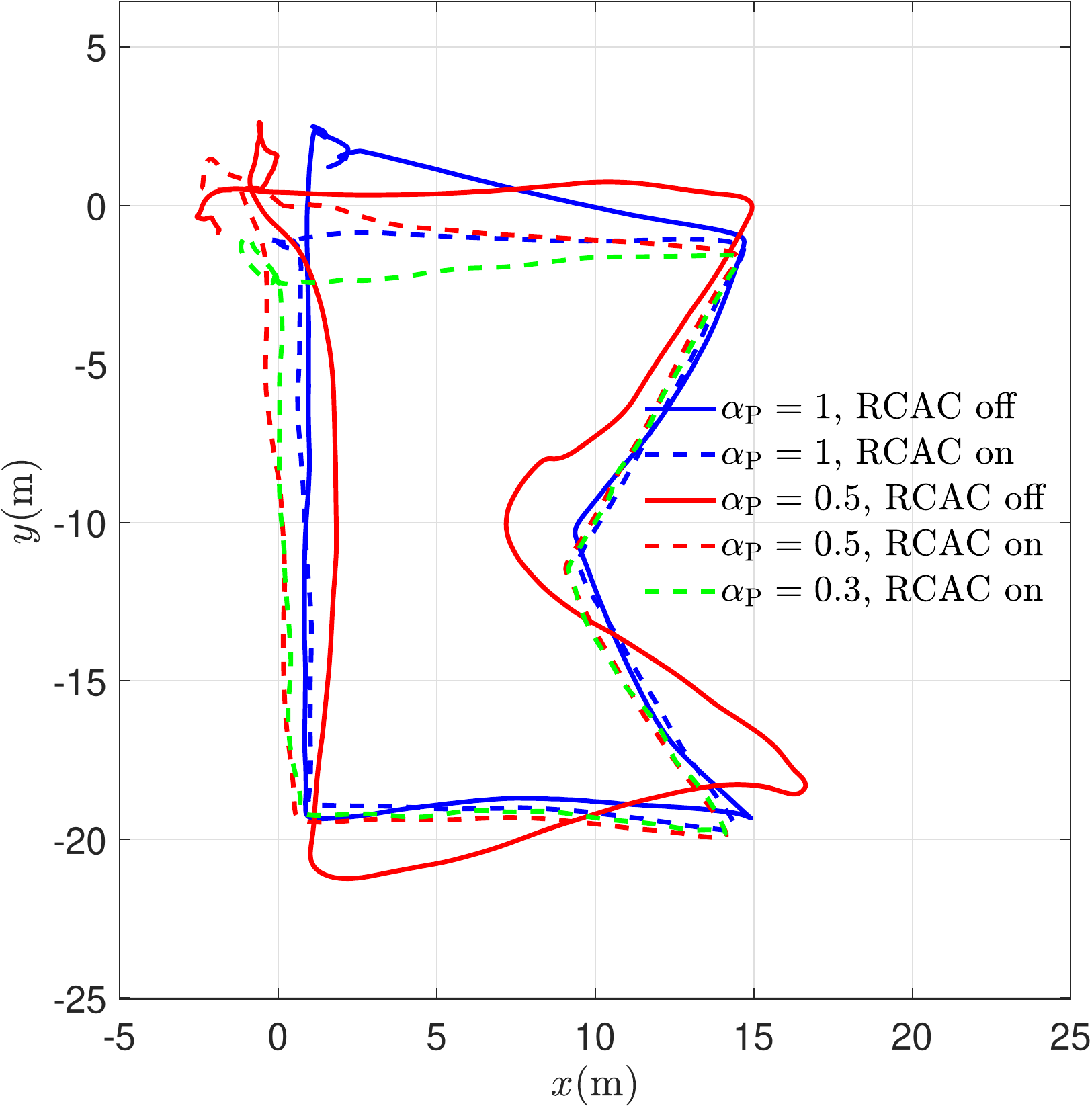}
    \caption{\footnotesize Closed-loop trajectory-following response of the Holybro X500 quadcopter. 
    The five traces show the baseline response, the degraded controller response, and the corresponding responses obtained with the adaptive PX4 autopilot. 
    }
    \label{fig:ACC21_PX4_HITL_traj_comp}
\end{figure}

\begin{figure}[h]
    \centering
    \includegraphics[width=.42\textwidth]{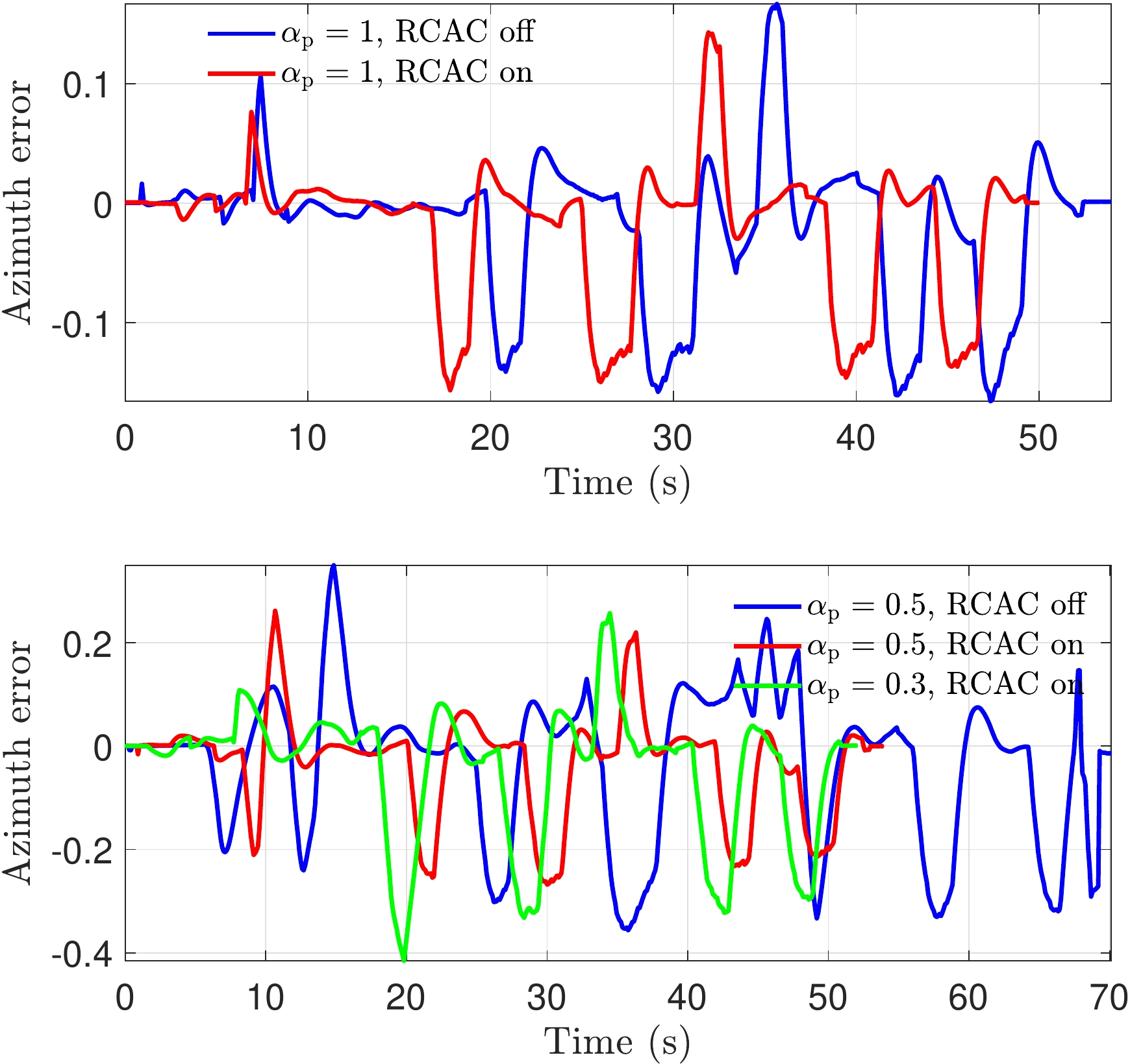}
    \caption{\footnotesize \footnotesize  Closed-loop azimuth response of the Holybro X500 quadcopter. 
    The blue trace shows the azimuth error without RCAC; and the
    red and the green traces show the azimuth error with RCAC for two values of $\alpha_\rmp.$
    Note that, in the case of $\alpha_\rmp=1,$
    the azimuth error is marginally better with RCAC.
    However, in the case of $\alpha_\rmp<1,$ RCAC optimizes the controller to improve the trajectory-following and the mission is completed in the same time taken in the baseline case. 
    }
    \label{fig:ACC21_PX4_HITL_RCAC_yaw_comp}
\end{figure}

\begin{figure}[h]
    \centering
    \includegraphics[width=.42\textwidth]{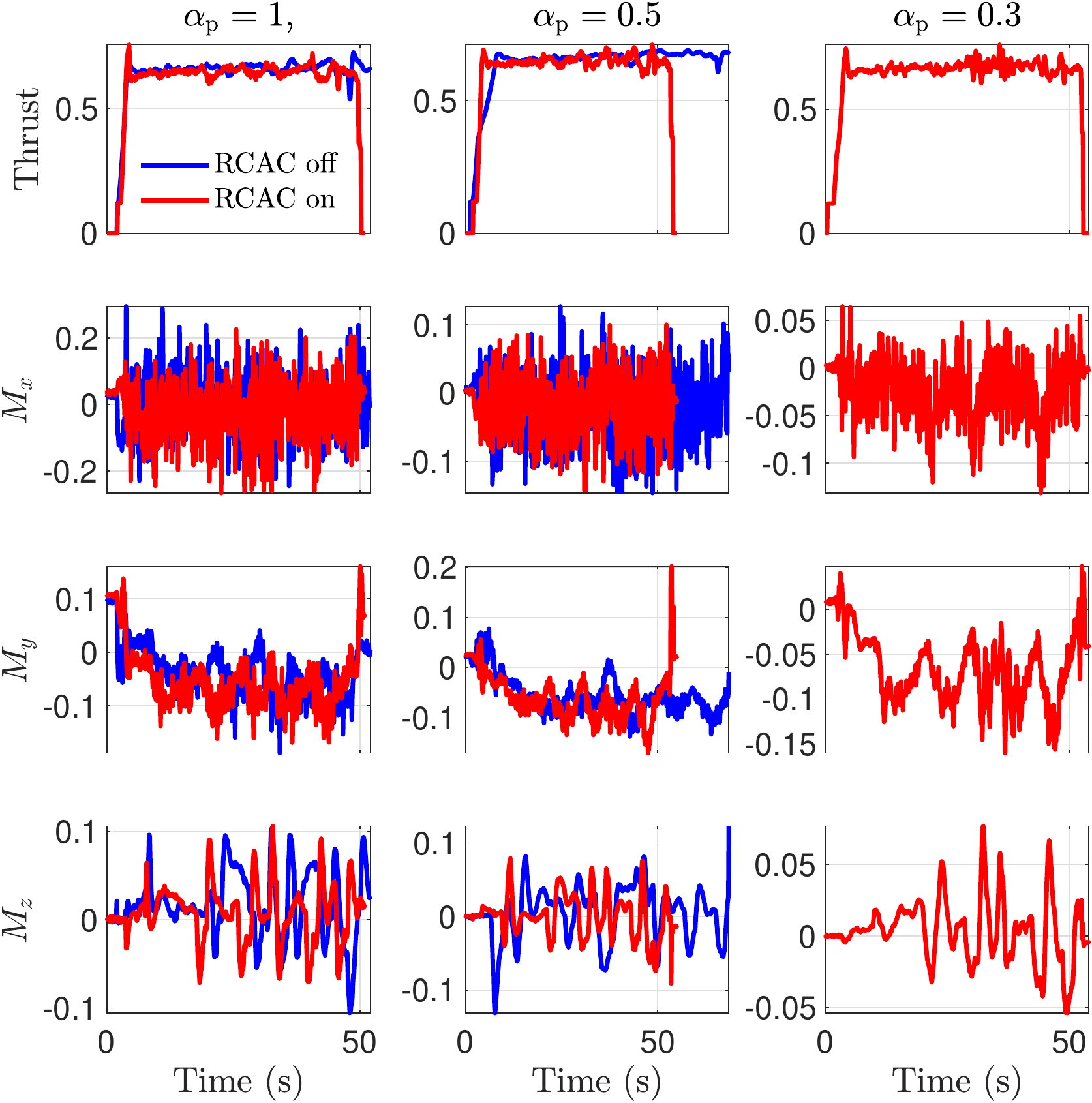}
    \caption{\footnotesize Thrust and the moment commands applied to the Holybro X500 quadcopter. 
    }
    \label{fig:ACC21_PX4_HITL_RCAC_actuator_comp}
\end{figure}

\begin{figure}[h]
    \centering
    \includegraphics[width=.42\textwidth]{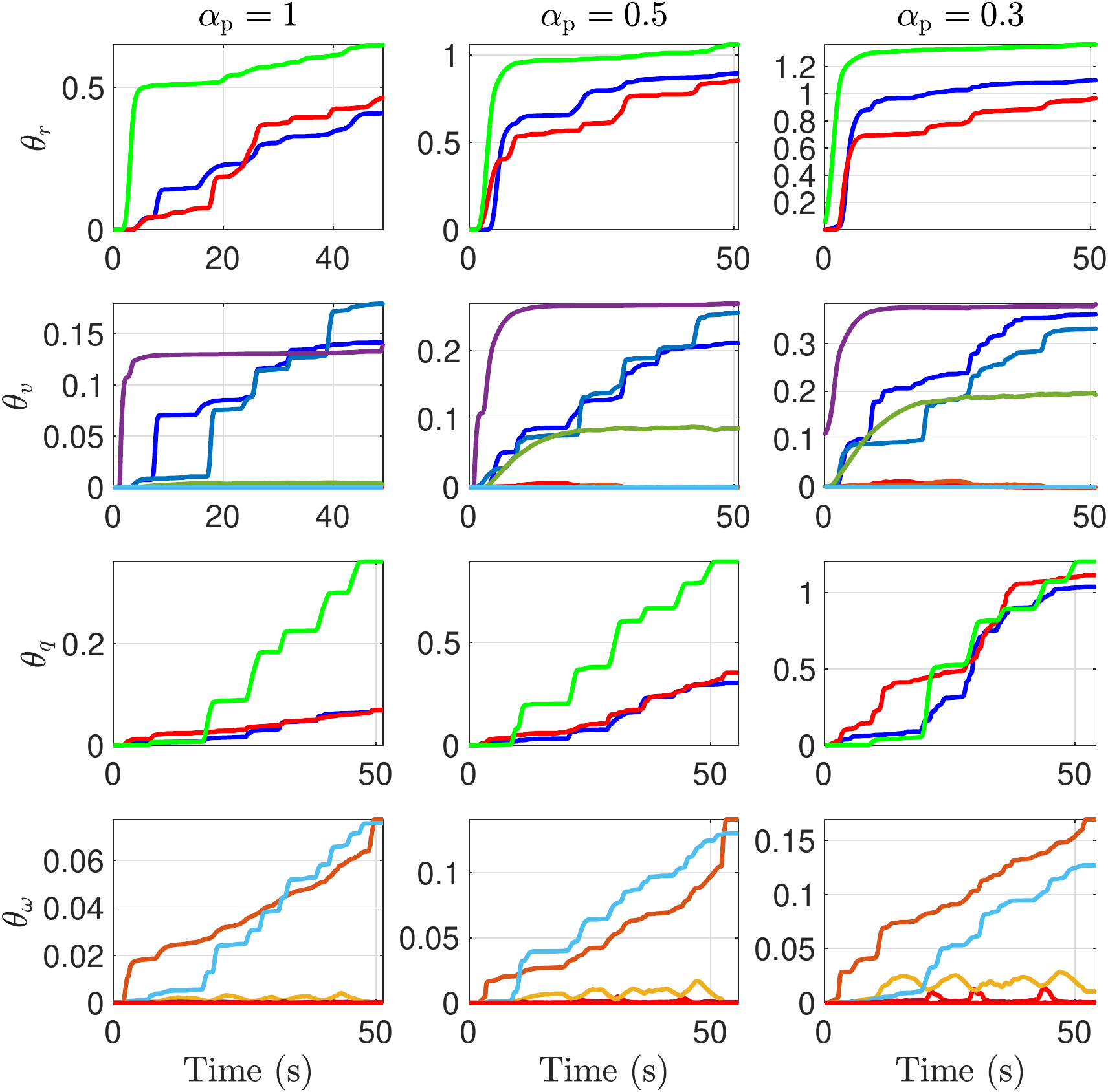}
    \caption{\footnotesize Adaptive gains optimized by RCAC in the adaptive digital autopilot. 
    Note that the magnitude of the gains increase as $\alpha_\rmp$ is reduced. 
    }
    \label{fig:ACC21_PX4_HITL_RCAC_theta}
\end{figure}

\section{Conclusions and Future Work}
\label{sec:conclusions}

This paper presented an adaptive digital autopilot that can improve an initial poor choice of controller gains. 
The adaptive autopilot is constructed by augmenting the fixed-gain controllers in the stock autopilot with adaptive controllers.
%
The adaptive autopilot was used to fly a quadcopter model in jMAVSim simulator and the X500 Holybro quadcoper. 
%
The adaptive autopilot recovered the performance in the case where the default controllers were degraded both in simulations and physical flight tests. 
Future work will focus on systematically assessing the sensitivity of the quadcopter performance to the four controllers in the PX4 autopilot,
targeting the augmentation to the most sensitive controller,
and 
investigating the performance improvements using variable rate forgetting, integrator anti-windup, and IIR position and attitude controllers. 
Furthermore, the performance of the adaptive autopilot under unknown suspended payload and chipped propellers, will be assessed.

\renewcommand*{\bibfont}{\small}
\printbibliography

\end{document}